\documentclass[11pt,letterpaper]{article}
\pdfoutput=1

\usepackage{amsmath,amsfonts,amssymb,amsthm}
\usepackage{amsthm}
\usepackage{graphicx,color}
\usepackage{boxedminipage}
\usepackage[ruled,boxed,linesnumbered]{algorithm2e}
\usepackage{framed}
\usepackage{thmtools}
\usepackage{thm-restate}
\usepackage{xspace}
\usepackage{todonotes}
\usepackage{footnote}
\usepackage{xcolor}
\usepackage{enumitem}
\usepackage{caption}
\usepackage{subcaption}

\usepackage[margin=1in,nohead]{geometry}

\usepackage{todonotes}
\usepackage{footnote}
 \usepackage[pdftex, plainpages = false, pdfpagelabels, 
                 bookmarks=false,
                 bookmarksopen = true,
                 bookmarksnumbered = true,
                 breaklinks = true,
                 linktocpage,
                 pagebackref,
                 colorlinks = true,  
                 linkcolor = blue,
                 urlcolor  = blue,
                 citecolor = red,
                 anchorcolor = green,
                 hyperindex = true,
                 hyperfigures
                 ]{hyperref} 
 \usepackage[nameinlink]{cleveref}
 \usepackage{xifthen}
 \usepackage{tabularx}
\usepackage{tikz} 
 \usetikzlibrary{calc,decorations.pathreplacing}

\tikzstyle{vertex}=[circle, fill=black!80!white, inner sep=0pt, minimum size=4pt]

\newtheorem{theorem}{Theorem}
\newtheorem{lemma}{Lemma}
\newtheorem{claim}{Claim}[section]

\newtheorem{definition}{Definition}
\newtheorem{observation}{Observation}
\newtheorem{proposition}{Proposition}

\theoremstyle{definition}

\newcommand{\cps}{{\sf compass}}

\DeclareMathOperator{\operatorClassFPT}{FPT\xspace}
\newcommand{\classFPT}{\ensuremath{\operatorClassFPT}\xspace}

\newlength{\RoundedBoxWidth}
\newsavebox{\GrayRoundedBox}
\newenvironment{GrayBox}[1]%
   {\setlength{\RoundedBoxWidth}{.93\textwidth}
    \def\boxheading{#1}
    \begin{lrbox}{\GrayRoundedBox}
       \begin{minipage}{\RoundedBoxWidth}}%
   {   \end{minipage}
    \end{lrbox}
    \begin{center}
    \begin{tikzpicture}%
       \node(Text)[draw=black!20,fill=white,rounded corners,%
             inner sep=2ex,text width=\RoundedBoxWidth]%
             {\usebox{\GrayRoundedBox}};
        \coordinate(x) at (current bounding box.north west);
        \node [draw=white,rectangle,inner sep=3pt,anchor=north west,fill=white] 
        at ($(x)+(6pt,.75em)$) {\boxheading};
    \end{tikzpicture}
    \end{center}}

\newenvironment{defproblemx}[2][]{\noindent\ignorespaces%
                                \FrameSep=6pt%
                                \parindent=0pt%
                \vspace*{-1.5em}
                \ifthenelse{\isempty{#1}}{%
                  \begin{GrayBox}{\textsc{#2}}%
                }{%
                  \begin{GrayBox}{\textsc{#2} parameterized by~{#1}}%
                }
                \begin{tabular*}{\textwidth}{@{\hspace{.1em}} >{\itshape} p{1.8cm} p{0.8\textwidth} @{}}%
            }{
                \end{tabular*}%
                \end{GrayBox}%
                \ignorespacesafterend
            }

\newcommand{\defproblema}[3]{
  \begin{defproblemx}{#1}
    Input:  & #2 \\
    Task: & #3
  \end{defproblemx}
}%

\newcommand{\pname}{\textsc}
\newcommand{\ProblemFormat}[1]{\pname{#1}}
\newcommand{\ProblemIndex}[1]{\index{problem!\ProblemFormat{#1}}}
\newcommand{\ProblemName}[1]{\ProblemFormat{#1}\ProblemIndex{#1}{}\xspace}

\newcommand{\probMaxSTP}{\ProblemName{Max Rank  $(s,t)$-Path}}

\newif\iflong
\begin{document}
\longtrue

\title{Computing paths of large rank in  planar frameworks deterministically\thanks{The research leading to these results has received funding from the Research Council of Norway via the project  BWCA (grant no. 314528). Giannos Stamoulis acknowledges support by the ANR project ESIGMA (ANR-17-CE23-0010) and the French-German Collaboration ANR/DFG Project UTMA (ANR-20-CE92-0027).}
}

\author{
Fedor V. Fomin\thanks{
Department of Informatics, University of Bergen, Norway.}
\and
Petr A. Golovach\addtocounter{footnote}{-1}\footnotemark{}
\and
Tuukka Korhonen\addtocounter{footnote}{-1}\footnotemark{}
\and 
Giannos Stamoulis\thanks{LIRMM, Universit\'e de Montpellier, CNRS, France.}
}

\date{}

\maketitle

\thispagestyle{empty}

\begin{abstract}
A framework consists of an undirected graph $G$ and a matroid $M$ whose elements correspond to the vertices of $G$. 
Recently, Fomin et al. [SODA 2023] and Eiben et al. [ArXiV 2023] developed parameterized algorithms for computing paths of rank $k$ in frameworks. More precisely, for vertices $s$ and $t$ of $G$, and an integer $k$, they gave FPT algorithms parameterized by $k$
deciding
whether there is an $(s,t)$-path in $G$ whose vertex set contains a subset of elements of $M$ of rank $k$. 
These algorithms are based on Schwartz-Zippel lemma for polynomial identity testing and thus are randomized,
and therefore the existence of a deterministic FPT algorithm for this problem remains open.

We present the first deterministic FPT algorithm that solves the problem in frameworks whose underlying graph $G$ is planar. While the running time of our algorithm is worse than the running times of the recent randomized algorithms, our algorithm works on more general classes of matroids. In particular, this is the first FPT algorithm  for the case when matroid $M$ is represented over rationals.

Our main technical contribution is the nontrivial adaptation of the classic irrelevant vertex technique to frameworks
to reduce the given instance to one of bounded treewidth. This allows us to employ the toolbox of representative sets to design a dynamic programming procedure solving the problem efficiently on instances of bounded treewidth. 
\end{abstract}
\newpage
\pagestyle{plain}
\setcounter{page}{1}

\section{Introduction}\label{sec:intro}
A \emph{framework} is a pair $(G,M)$, where $G$ is a graph and $M=(V(G),\mathcal{I})$ is a matroid on the vertex set of $G$.
This term appears in the recent monograph of Lov{\'{a}}sz~\cite{Lovasz19}, where he defines frameworks as graphs with a collection of vectors of $\mathbb{R}^d$ labeling their vertices.
Frameworks have appeared in the literature under many different names.
For example, they are mentioned as \emph{pregeometric graphs} in the influential work of Lov{\'{a}}sz~\cite{Lovasz77} 
on representative families of linear matroids and as \emph{matroid graphs}
in the book of Lov{\'{a}}sz and Plummer \cite{LovaszPlummerbook876}.
The problem of computing maximum matching in frameworks is closely related to the matchoid, the matroid parity, and polymatroid matching problems (see~\cite{LovaszPlummerbook876} for an overview).
More broadly, the problems of finding specific subgraphs of large ranks in frameworks belong to the wide family of problems about submodular function optimization under combinatorial constraints \cite{CalinescuCPV11,ChekuriP05,NemhauserWF78,FominGKLS23shor}.

Fomin et al. in~\cite{FominGKSS23fixe} introduced the following  \textsc{Maximum Rank $(s,t)$-Path} problem.  In this problem, 
 given a framework $(G,M)$, two vertices $s$ and $t$ of $G$, and an integer $k$,
we seek for an $(s,t)$-path in $G$ where the rank function of $M$ evaluates to at least $k$. 
We say that such a path \emph{has rank at least $k$}.

\defproblema{\probMaxSTP}%
{A framework $(G,M)$, vertices $s$ and $t$ of $G$, and an  integer $k\geq 0$.}%
{Decide whether $G$ contains an $(s,t)$-path of rank at least $k$.}

\probMaxSTP encompasses several fundamental and well-studied problems about paths and cycles in undirected graphs.

\medskip
\noindent\emph{Longest path}. Of course, when $M$ is a uniform matroid, then a path is of rank at least $k$ if and only if it contains at least $k$ vertices. In this case, we have the classical \textsc{Longest Path} problem, where for a graph $G$ and integer $k$ the task is to identify whether $G$ contains a path with at least $k$ vertices~\cite{AlonYZ95}.

\medskip
\noindent\emph{$T$-cycle}.  In this problem, we are given a set $T$ of terminals and the task is to decide whether there is a cycle through all terminals \cite{BjorklundHT12,Kawarabayashi08,Wahlstrom13}. {$T$-cycle} is the special case of \probMaxSTP.
 Consider the following linear matroid. For every vertex of $G$ not in $T$ we assign a $|T|$-dimensional vector whose all entries are zero. To vertices of $T$ we assign vectors forming an orthonormal basis of $\mathbb{R}^{|T|}$. Then $G$ has a cycle passing through all terminals if and only if $(G,M)$ has an $(s,t)$-path of rank $|T|$, for some $\{s,t\}\in E(G)$.

\medskip
\noindent\emph{Maximum Colored Path.}  
In the  \textsc{Maximum Colored $(s,t)$-Path} problem, we are given a colored graph $G$, two vertices $s$ and $t$ of $G$, and an integer $k$. The task is to decide whether $G$ has an $(s,t)$-path containing at least $k$ different colors~\cite{broersma2005paths,FominGKSS23fixe} (see also~\cite{CohenIMTP21,CouetouxNV17}). \textsc{Maximum Colored $(s,t)$-Path} is the special case of \probMaxSTP where the matroid $M$ is a partition matroid. Indeed, in this matroid the ground set $V(G)$ is partitioned into classes $L_1, \dots, L_t$ and a set $I$ is independent if $|I\cap L_i|\leq 1$ for every label $i\in \{1,\dots, t\}$.
In this way, a path of $G$ of rank at least $k$ is a path containing vertices of at least $k$ different (color) classes among $L_1, \dots, L_t$.

\paragraph{Randomized FPT algorithms for \textsc{Maximum Rank $(s,t)$-Path}.}
The parameterized complexity of \textsc{Maximum Rank $(s,t)$-Path} was unknown until very recently.
The first FPT algorithm for \textsc{Maximum Rank $(s,t)$-Path} was given in~\cite{FominGKSS23fixe}.
This algorithm runs in time $2^{\mathcal{O}(k^2\log(q+k))}n^{\mathcal{O}(1)}$ and works on frameworks with matroids represented in finite fields of order $q$.
Also, Eiben, Koana, and Wahlström~\cite{EibenKW23dete}, using different techniques, obtained an FPT algorithm for the same problem that runs in time $2^k n^{\mathcal{O}(1)}$ on frameworks with matroids representable over fields of characteristic two.
These two algorithms use two different algebraic methods.
The algorithm of~\cite{FominGKSS23fixe} extends the celebrated algebraic technique based on \emph{cancellation of monomials} used by Björklund, Husfeldt, and Taslaman~\cite{BjorklundHT12} to solve the \textsc{$T$-Cycle} problem, while the algorithm of~\cite{EibenKW23dete} utilizes the toolbox of \emph{(constrained) multilinear detection}~\cite{Koutis08fast,KoutisW15alge,DBLP:journals/siamcomp/Bjorklund14,BjorklundHKK17} 
combined with \emph{determinantal sieving}~\cite{EibenKW23dete}.
Both these algorithms involve polynomial identity testing and invoke the Schwartz-Zippel lemma, and therefore are randomized.
In fact, because of the crucial use of the Schwartz-Zippel lemma in both these algorithms, as the authors of~\cite{EibenKW23dete} state it, ``derandomization appears infeasible" for the algorithms of~\cite{FominGKSS23fixe} and~\cite{EibenKW23dete} for  \textsc{Maximum Rank $(s,t)$-Path}.
Therefore, the next challenge is to obtain \textit{derandomized} FPT algorithms for this problem.

\paragraph{Our results.}
Our main result establishes the first \textit{deterministic} FPT algorithm for \textsc{Maximum Rank $(s,t)$-Path} on frameworks of planar graphs and matroids representable over finite fields or  over the field of rationals.

 \begin{restatable}{theorem}{thmirrelevantv}
\label{thm:fpt-planar}
There is a deterministic algorithm that, given a framework $(G,M)$, where $G$ is a planar graph $G$ and $M$ is represented as a matrix over a finite field or over $\mathbb{Q}$, two vertices $s,t\in V(G)$ and an integer $k$, in time $2^{2^{\mathcal{O}(k\log k)}}\cdot (|G|+\|M\|)^{\mathcal{O}(1)}$ either returns an $(s,t)$-path of $G$ of rank at least $k$, or determines that $G$ has no such $(s,t)$-path.
\end{restatable}
Note that the randomized FPT algorithms of~\cite{FominGKSS23fixe} and~\cite{EibenKW23dete} work for matroids representable over finite fields or fields of characteristic two.
The algorithm of~\Cref{thm:fpt-planar}, apart from being the first deterministic algorithm for \textsc{Maximum Rank $(s,t)$-Path},
is also the first FPT algorithm for frameworks whose matroids are \textsl{not} represented over a finite field or a field of characteristic two, but are represented over $\mathbb{Q}$.

\paragraph{Our techniques.}
To design the deterministic FPT algorithm of~\Cref{thm:fpt-planar}, we follow a different proof strategy than that of~\cite{FominGKSS23fixe} and~\cite{EibenKW23dete}.
Our approach is based on the \textsl{win/win} arguments of the celebrated \emph{irrelevant vertex technique} of Robertson and Seymour~\cite{RobertsonS95b}.
The general scheme of this technique is the following.
If the graph satisfies certain combinatorial properties,
then one can identify a vertex of the graph that can be declared \emph{irrelevant}, meaning that its deletion results in an equivalent instance of the problem.
Therefore, after deleting this vertex, we can iterate on the (equivalent) reduced instance.
Once this reduction rule can not be further applied, the obtained reduced instance is equivalent to the original one and also ``simpler''.
Therefore, one remains to argue that the problem can be solved efficiently in the reduced equivalent instance.
This is a standard technique in parameterized algorithms design -- see, for example,~\cite{FominGT19modif,JansenK021verte,KawarabayashiMR08asim,MarxS07obta,KobayashiK09algo,GroheKMW11find,KawarabayashiR09hadw,FominLST12line,Kawarabayashi08,KawarabayashiKM10link,KawarabayashiR10oddc,Kawarabayashi09plan,FominLP0Z20hittin,BasteST20acomp,SauST21kapiII,GolovachST23model} (see also~\cite[Section 7.8]{cygan2015parameterized}).
The standard mesure of complexity of instances for the application of the irrelevant vertex technique is \emph{treewidth}.
In particular, the strategy is formulated as follows. As long as the treewidth of the instance is large enough, detect and remove irrelevant vertices.
If the treewidth is small, then solve the problem on this equivalent instance using dynamic programming.

Our application of the irrelevant vertex technique is inspired by the algorithm of Kawa\-ra\-bayashi~\cite{Kawarabayashi08} for \textsc{$T$-cycle}
and extends its methods.
In a typical irrelevant-vertex argument, one has to prove that every solution can ``avoid'' a vertex that will be declared irrelevant.
For example, in the classical application of Robertson and Seymour~\cite{RobertsonS95b} for the \textsc{Disjoint Paths} problem, one should argue that (if the graph has large treewidth) any collection of disjoint paths between certain terminals can be ``rerouted away'' from a vertex $v$ and this vertex should be declared irrelevant.
In our case, where we seek an $(s,t)$-path of \emph{large rank} in a framework, this rerouting should guarantee that large rank is preserved.
In general, to deal with such problems on frameworks, one should employ new arguments to adjust this technique to take into account the structure of the matroid.
The way we circumvent this problem for \textsc{Maximum Rank $(s,t)$-Path} is to formulate such a rerouting argument in a ``sufficiently insulated'' area of the graph where independent sets of the matroid $M$ appear in a homogeneous way.
Planarity of the input graph allows to find such an area using the grid-like structure of \emph{walls}.
An overview of this approach is provided in~\Cref{subsec:overview}.
This application of the irrelevant vertex technique for frameworks is novel and illustrates an interesting interplay between combinatorial structures and algebraic properties, that may be of independent interest.

The dynamic programming on graphs of bounded treewidth is pretty standard (see, e.g.,~\cite{CyganFKLMPPS15}) up to one detail. 
To encode a partial solution, we keep the information about vertices forming independent sets of matroid $M$ visited by a partial solution.
However,  the number of independent sets of size at most $k$ in $M$ could be of order $n^k$.  Thus a naive encoding of partial solutions would result in blowing-up of the computational complexity.  To avoid this, we store only \emph{representative} sets (see~\cite{FominLPS16,LokshtanovMPS18}) instead of all possible independent sets.
Both randomized~\cite{FominLPS16} and deterministic~\cite{LokshtanovMPS18} constructions of representative sets require a linear representation of $M$. This is the reason why \Cref{thm:fpt-planar} is stated for linear matroids. We point out that the dynamic programming subroutine for graphs of bounded treewidth is the only place in the proof of \Cref{thm:fpt-planar} requiring a representation of $M$. It is an interesting open question, whether  \textsc{Maximum Rank $(s,t)$-Path} is \classFPT when parameterized by $k$ and the treewidth if the input matroid is given by its independence oracle. 

 \subsection{Overview of the proof of \Cref{thm:fpt-planar}}
 \label{subsec:overview}

Our general approach is the following. We show that if the {treewidth} of the input graph $G$ is $2^{\mathcal{O}(k\log k)}$, then \textsc{Maximum Rank $(s,t)$-Path} can be solved in \classFPT time by a dynamic programming algorithm. Otherwise, if the treewidth is sufficiently large,  we give an algorithm that either finds an $(s,t)$-path of rank at least $k$ or identifies 
 an \emph{irrelevant} vertex $v$, that is, a vertex whose deletion results in an equivalent instance of the problem. In the latter case, we delete $v$ and iterate on the reduced instance. 

If the treewidth of the input graph is large, i.e., of order $2^{\Omega(k\log k)}$, we exploit the grid-minor theorem of Robertson and Seymour for planar graphs~\cite{RobertsonST94} that asserts that a planar graph either contains $(w\times w)$-grid as a minor or the treewidth is $\mathcal{O}(w)$. More precisely, we have that given a plane embedding of $G$, we can find a plane 
 $h$-wall for $h=2^{\Omega(k\log k)}$ as a topological minor or, equivalently, a plane subgraph of $G$ that is a subdivision of such a wall. To explain our arguments, we need some notions that are informally explained here by making use of figures.  In particular, an example of an $h$-wall for $h=7$ is given in \Cref{fig:wall}.
  \begin{figure}[ht]
\centering
\begin{tikzpicture}[scale=0.6]
\def\l{14} 
\def\lminus{13} 
\def\h{7} 
\def\hminus{6} 
\foreach \i in {1,2,...,\l}{
\foreach \j in {1,2,...,\h}{
   \ifthenelse{\(\not \i=\l\) \AND \not \(\i=\lminus \and \j = \h\) \AND \not \(\i=1 \and \j = 1\)}
   {\draw[black!50!white,line width=0.8pt] (\i,\j) -- (\i+1,\j)}
   {};
   \ifthenelse{\(\isodd{\j} \AND \not \(\isodd{\i}\) \AND \not \j=\h \)\OR \(\isodd{\i} \AND \not \(\isodd{\j}\)\)}
  {\draw[black!50!white,line width=0.8pt] (\i,\j) -- (\i,\j+1)}
  {}
  ;}
 }

\foreach \i in {2,...,\lminus}{
\draw[red,line width=1.7pt] (\i,1) -- (\i+1,1) (\i-1,\h) -- (\i,\h);
}

\foreach \j in {2,4,...,\hminus}{
\draw[red,line width=1.7pt] (1,\j) -- (1,\j+1) (1,\j) -- (2,\j) (1,\j+1) -- (2,\j+1) (\lminus,\j) -- (\lminus,\j+1) (2,\j-1) -- (2,\j) (\l,\j-1) -- (\l,\j) (\l-1,\j) -- (\l,\j) (\l,\j-1) -- (\l-1,\j-1);
}

\foreach \i in {6,...,\lminus}{
\draw[blue!80!green,line width=1.9pt] (\i-3,2) -- (\i-2,2) (\i-2,\h-1) -- (\i-1,\h-1);
}

\foreach \j in {3,5,...,\hminus}{
\draw[blue!80!green,line width=1.9pt] (3,\j-1) -- (4,\j-1) (3,\j-1) -- (3,\j) (3,\j) -- (4,\j) (4,\j) -- (4,\j+1)
 (\l-2,\j) -- (\l-2,\j+1) (\l-3,\j-1) -- (\l-3,\j) (\l-2,\j) -- (\l-3,\j) (\l-2,\j+1) -- (\l-3,\j+1);
}

\foreach \i in {6,7,...,9}{
\draw[blue!80!green,line width=1.9pt] (\i,3) -- (\i+1,3) (\i,5) -- (\i-1,5);
}

\draw[blue!80!green,line width=1.8pt] (6,3) -- (6,4) -- (5,4) -- (5,5) (9,5) -- (9,4) -- (10,4) -- (10,3);

\foreach \i in {1,2,...,\l}{
\foreach \j in {1,2,...,\h}{
  \ifthenelse{\(\i=\l \and \j = \h\) \OR \(\i=1 \and \j = 1\)}
  {}
  {\node[vertex] (A\i\j) at (\i,\j) {}};
  }
}

\end{tikzpicture}
\caption{A $7$-wall and its layers.}
\label{fig:wall}
\end{figure}
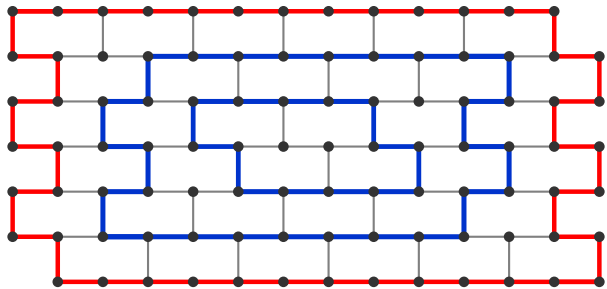 
 
 Note that an $h$-wall has 
 $\lfloor h/2\rfloor$ nested cycles, called \emph{layers}, that are shown in \Cref{fig:wall} in red and blue. The layer forming the boundary of a wall is called the \emph{perimeter} of the wall and is shown in red in the figure.  We extend the notions of layers and perimeter for a \emph{subdivided} $h$-wall, that is, the graph obtained from an $h$-wall by replacing some of its edges by paths. Given a plane subdivided $h$-wall $W$ in $G$, we call the subgraph of $G$ induced by the vertices on the perimeter and inside the inner face of the perimeter the \emph{compass} of $W$ and denote it by $\cps(W)$.   Notice that we can assume that the compass of the subdivided $h$-wall $W$ in $G$ does not contain the terminal vertices $s$ and $t$ by switching to a smaller subwall if necessary. Furthermore, we can assume that $\cps(W)$ is a 2-connected graph as any $(s,t)$-path can only contain vertices of the biconnected component of $\cps(W)$ containing $W$. Also we can assume that  $G$ has two disjoint paths connecting $s$ and $t$ with two distinct vertices on the perimeter of $W$; otherwise, any vertex of $\cps(W)$ outside the perimeter is trivially irrelevant.  
 
Observe that for any nontrivial subwall $W'$ of $W$, $\cps(W')$ is also 2-connected. Therefore, for every two distinct vertices $x$ and $y$ on the perimeter of $W'$ and any $z\in V(\cps(W'))$, $\cps(W')$ has internally disjoint $(x,z)$ and $(y,z)$-paths. In particular, given a set of vertices $S\subseteq V(\cps(W'))$ that are independent with respect to $M$, we can join any $z\in S$ with $x$ and $y$ by disjoint paths.  This observation is crucial for us. 

 \begin{figure}[h!]
\centering
\includegraphics[scale=0.7]{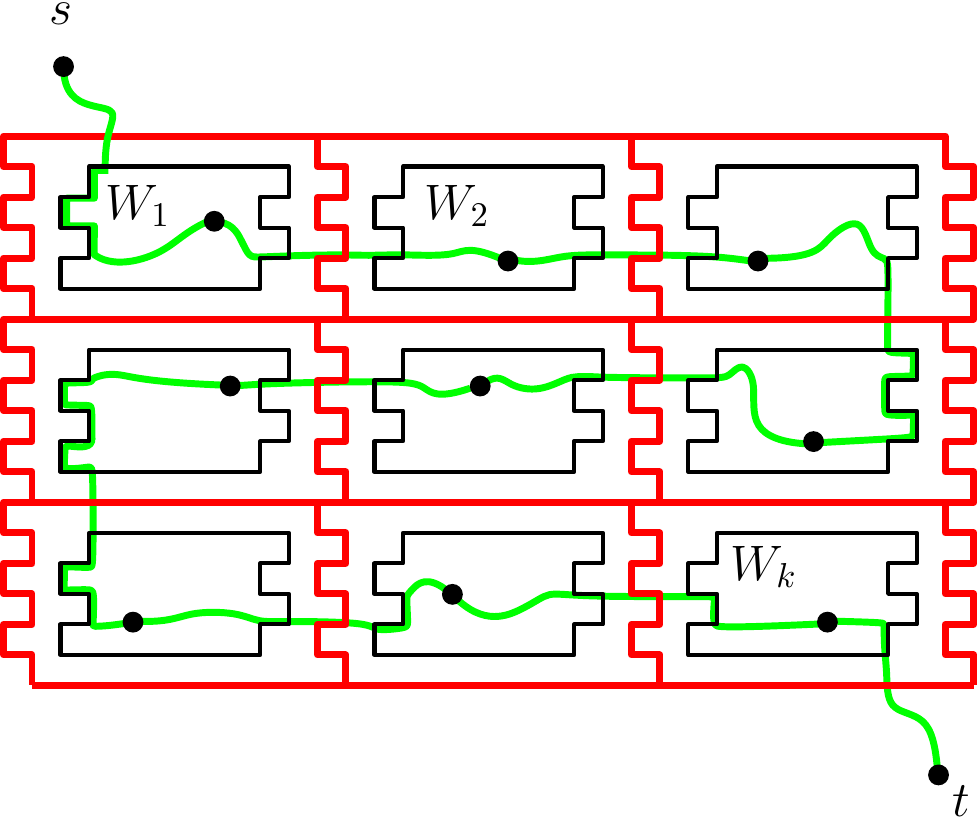}
\caption{An $(s,t)$-path for walls of big rank.}
\label{fig:path}
\end{figure} 

Suppose that there is a packing of $k$ subwalls $W_1,\ldots,W_k$ in $W$ separated by paths in $W$ as it is shown in \Cref{fig:path} such that the rank  $r(\cps(W_i))\geq k$ for $i\in \{1,\ldots,k\}$. Then we can choose vertices $v_1,\ldots,v_k$ in $\cps(W_1),\ldots,\cps(W_k)$, respectively, in such a way that $\{v_1,\ldots,v_k\}$ is an independent set of $M$. Then by our observation, we can construct an $(s,t)$-path in $G$ that goes through $v_1,\ldots,v_k$ as it is shown in the figure in green. Suppose that this is not the case. Then, by zooming inside the wall, we can assume that $r(\cps(W))<k$. Moreover, by recursive zooming, we can find a subwall $W'$ of $W$ with the following structural properties (see \Cref{fig:reroute}).
\begin{itemize}
\item There is a packing of $k+1$ subwalls $W_0,W_1,\ldots,W_k$ in $W'$ separated by paths in $W'$ shown in red in \Cref{fig:reroute} such that $r(\cps(W_i))=r(\cps(W'))$ for $i\in \{1,\ldots,k\}$. 
\item The packing of $W_0,W_1,\ldots,W_k$ is surrounded by $\mathcal{O}(k^2)$ ``insulation'' layers of $W'$ shown in blue.
\end{itemize}
We claim that vertices of $W_0$ are irrelevant.

To see this, consider an $(s,t)$-path $P$ of rank at least $k$ in $G$. We show that if $P$ goes through a vertex of $W_0$, then the path can be rerouted as it is shown in \Cref{fig:reroute} in green to avoid $W_0$. Consider an independent set $X\subseteq V(P)$ of rank $k$ and let $u_1,\ldots,u_\ell$ be the vertices of $X$ that are not spanned by $V(\cps(W'))$ in $M$. Then $u_1,\ldots,u_\ell$ are outside $W'$. We prove that there are two distinct vertices $x$ and $y$ on the inner insulation layer of $W'$, and an $(s,x)$-path $P_1$ and an $(y,t)$-path $P_2$ such that (i) $x$ and $y$ are unique vertices of these paths in the inner insulation layer, and (ii) $u_1,\ldots,u_\ell\in V(P_1)\cup V(V_2)$. 
The proof that $\mathcal{O}(k^2)$ insulation layers are sufficient for rerouting $P$ is non-trivial.  In particular, we adapt the ideas from~\cite{Kawarabayashi08} as well as the structural results of Kleinberg~\cite{Kleinberg98}. Further, we show that it is possible to choose vertices $v_1,\ldots,v_k$ in $W_1,\ldots,W_k$, respectively, so that $r(\{v_1,\ldots,v_k\})=r(\cps(W'))$. Then we construct an $(x,y)$-path $Q$ in the inner part of $W'$ such that (i) $Q$ is internally disjoint with $P_1$ and $P_2$, (ii) $Q$ goes through $v_1,\ldots,v_k$, and (iii) $Q$ avoids $W_0$. We have that $P'=P_1QP_2$ is an $(s,t)$-path that goes through $u_1,\ldots,u_\ell$ and   $v_1,\ldots,v_k$. Therefore $r(P')\geq r(X)\geq k$. Since $Q$ avoids $W_0$, $P'$ has the same property. 

 \begin{figure}[h!]
\centering
\includegraphics[scale=0.65]{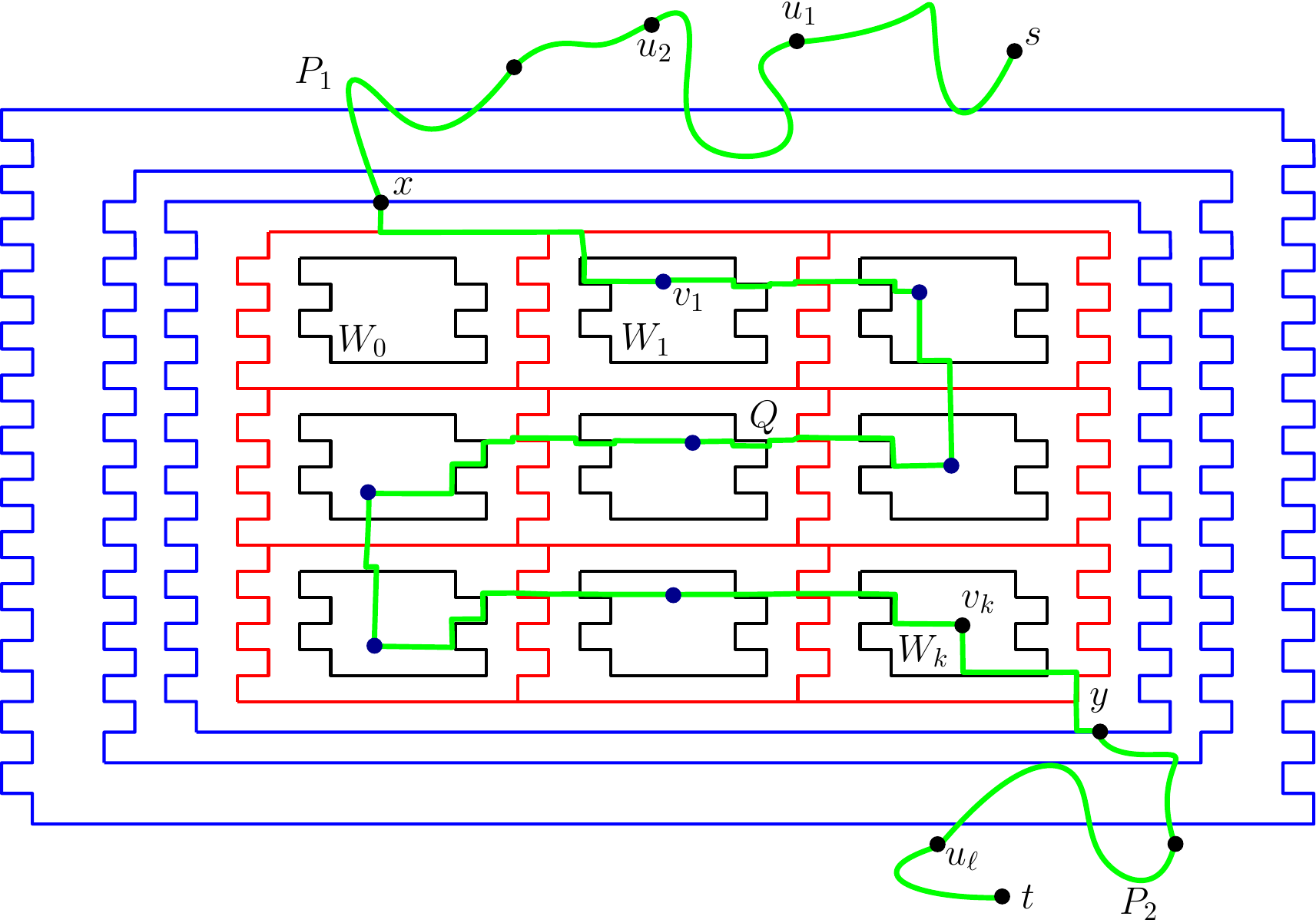}
\caption{Rerouting an $(s,t)$-path.}
\label{fig:reroute}
\end{figure} 

Finally, we note that the algorithm of  Kawarabayashi~\cite{Kawarabayashi08} for \textsc{$T$-cycle} works for general graphs.  
The statement of  \Cref{thm:fpt-planar} is limited to planar graphs and planarity is required to 
  ensure that the rerouting does not decrease the rank of an $(s,t)$-path. It is quite plausible that with additional technicalities our method could be lifted when the underlying graph of the framework is of bounded genus, and more generally,  minor-free. However, it is very unclear, whether rerouting that does not decrease the rank could be achieved for general graphs.  It remains the main obstacle towards 
pushing  the irrelevant vertex technique from frameworks with planar graphs to frameworks with general graphs.

\paragraph{Organization of the paper.}
In~\Cref{sec:prelim}, we present some basic definitions and preliminary results.
In~\Cref{sec:det} we show how to reduce to instances of bounded treewidth using the irrelevant vertex technique,
while in~\Cref{subsec:dp-planar} we present the dynamic programming algorithm that solves the problem in instances of bounded treewidth.
We conclude in~\Cref{sec:conclusion} with open questions and possible future research directions.

\section{Preliminaries}\label{sec:prelim} 
In this section, we introduce basic notation and state some auxiliary results. In~\Cref{subsec:basic}, we provide some basic definitions on parameterized complexity and on graphs, while in~\Cref{subsec:walls-treewidth}, we give some necessary definitions and results on walls and treewidth. We conclude this section with~\Cref{subsec:frameworks}, where we provide some useful notions on matroids and the definition of frameworks.

\subsection{Basic definitions}\label{subsec:basic}

We use $\mathbb{Z}_{\ge 1}$ to denote the set of positive integers and $\mathbb{Z}_{\ge 0}$ the set of non-negative integers. Also, given integers $p,q$ such that $p<q$, we use $[p,q]$ to denote the set $\{p,p+1,\ldots,q\}$ and, if $p\geq1$, we use $[p]$ to denote the set $\{1,\ldots,p\}$.

\paragraph{Parameterized Complexity.} We refer to the book of Cygan et al.~\cite{cygan2015parameterized} for introduction to the area. Here we only briefly mention the notions that are most important to state our results.  A \emph{parameterized problem} is a language $L\subseteq\Sigma^*\times\mathbb{N}$, where $\Sigma^*$ is a set of strings over a finite alphabet $\Sigma$. An input of a parameterized problem is a pair $(x,k)$, where $x$ is a string over $\Sigma$ and $k\in \mathbb{N}$ is a \emph{parameter}. 
A parameterized problem is \emph{fixed-parameter tractable} (or \classFPT) if it can be solved in time $f(k)\cdot |x|^{\mathcal{O}(1)}$ for some computable function~$f$.  
The complexity class \classFPT contains  all fixed-parameter tractable parameterized problems.

\paragraph{Graphs.}
We use standard graph-theoretic terminology and refer to the textbook of Diestel~\cite{Diestel12} for missing notions.
We consider only finite  graphs, and the considered graphs are assumed to be undirected if it is not explicitly said to be otherwise.  For  a graph $G$,  $V(G)$ and $E(G)$ are used to denote its vertex and edge sets, respectively. 
Throughout the paper we use $|G|=|V(G)|$.
For a graph $G$ and a subset $X\subseteq V(G)$ of vertices, we write $G[X]$ to denote the subgraph of $G$ induced by $X$. 
For a vertex $v$, we denote by $N_G(v)$ the \emph{(open) neighborhood} of $v$, i.e., the set of vertices that are adjacent to $v$ in $G$. For $X\subseteq V(G)$, $N_G(X)=\big(\bigcup_{v\in X}N_G(v)\big)\setminus X$.
The \emph{degree} of a vertex $v$ is $d_G(v)=|N_G(v)|$. 

A walk $W$ of length $\ell$ in $G$ is a sequence of vertices $v_1, v_2, \ldots, v_\ell$, where $v_i v_{i+1} \in E(G)$ for all $1 \le i < \ell$.
The vertices $v_1$ and $v_\ell$ are the \emph{endpoints} of $W$ and the vertices $v_2,\ldots,v_{\ell-1}$ are the \emph{internal} vertices of $W$.
A path is a walk where no vertex is repeated.
For a path $P$ with endpoints $s$ and $t$, we say that $P$ is an $(s,t)$-path. 
A cycle is a path with the additional property that $v_\ell v_1 \in E(G)$ and $\ell \ge 3$.

\subsection{Walls and treewidth}\label{subsec:walls-treewidth}

\paragraph*{Walls.}
Let  $k,r\in\mathbb{N}.$ The
\emph{$(k\times r)$-grid} is the
graph whose vertex set is $\{1,\ldots,k\}\times\{1,\ldots,r\}$ and two vertices $(i,j)$ and $(i',j')$ are adjacent if and only if $|i-i'|+|j-j'|=1.$
An  \emph{elementary $r$-wall}, for some odd integer $r\geq 3,$ is the graph obtained from a
$(2 r\times r)$-grid
with vertices $(x,y)\in\{1,\ldots,2r\}\times\{1,\ldots,r\},$
after the removal of the
``vertical'' edges $\{(x,y),(x,y+1)\}$ for odd $x+y,$ and then the removal of 
all vertices of degree one.
Notice that, as $r\geq 3,$  an elementary $r$-wall is a planar graph
that has a unique (up to topological isomorphism) embedding in the plane such that all its finite faces are incident to exactly six
edges.
The {\em perimeter} of an elementary $r$-wall is the cycle bounding its infinite face.

An {\em $r$-wall} is any graph $W$ obtained from an elementary $r$-wall $\bar{W}$
after subdividing edges.
We call the vertices that where added after the subdivision operations {\em subdivision vertices},
while we call the rest of the vertices (i.e., those of $\bar{W}$) {\em branch vertices}.
The {\em perimeter} of $W$, denoted by $\mathsf{perim}(W)$, is the cycle of $W$ whose non-subdivision vertices are the vertices of the perimeter of $\bar{W}$.
A {\em subdivided edge} of $W$ is a path of $W$ whose endpoints are two branch
vertices of $W$ and its internal vertices are subdivision vertices of $W$.
We also call a vertex $v$ an {\em in-peg} of the perimeter of $W$, if $v\in V(\mathsf{perim}(W))$ and $v$ has degree three in $W$.

 A graph $W$ is a {\em wall} if it is an $r$-wall for some odd $r\geq 3$
and we refer to $r$ as the {\em height} of $W.$ Given a graph $G,$
a {\em wall of} $G$ is a subgraph of $G$ that is a wall.
We insist that, for every $r$-wall, the number $r$ is always odd.
Let $W$ be a wall of a graph $G$ and $K'$ be the connected component of $G\setminus V(\mathsf{perim}(W))$ that contains $W\setminus V(\mathsf{perim}(W))$.
We use $\mathsf{inn}(W)$ to denote the graph $K'$.
The {\em compass} of $W$, denoted by $\cps(W)$, is the graph $G[V(\mathsf{inn}(W))\cup V(\mathsf{perim}(W))]$.

The {\em layers} of an $r$-wall $W$, for any odd integer $r\geq 3$, are recursively defined as follows.
The first layer of $W$ is its perimeter.
For $i=2,\ldots, (r-1)/2$, the $i$-th layer of $W$ is the $(i-1)$-th layer of the wall $W'$ obtained from $W$ after removing from $W$ its perimeter and all occurring vertices of
degree one.
Notice that each $(2k+1)$-wall has $k$ layers.
For every $i=1,\ldots, (r-1)/2$, we use $L_i$ to denote the $i$-th layer of $W$.
Also,  $i=2,\ldots, (r-1)/2$ we use $W^{(i)}$ to denote the wall obtained from $W$ after removing from $W$ the layers $L_1,\ldots, L_i$ and all occurring vertices of degree one and we set $W^{(1)} := W$.
Notice that for every $i=1,\ldots, (r-1)/2$, $\mathsf{perim}(W^{(i)}) = L_i$.
See~\autoref{fig:wall} for an example.

\paragraph*{Treewidth.} 
A \emph{tree decomposition} of a graph~$G$
is a pair~$(T,\mathcal{X})$ where $T$ is a tree and $\mathcal{X} = \{X_i \mid i\in V(T)\}$ is a family of subsets of $V(G)$
such that
\begin{itemize}
\item $\bigcup_{t \in V(T)} X_t = V(G),$
\item for every edge~$e$ of~$G$ there is a $t\in V(T)$ such that $X_t$ contains both endpoints of~$e,$ and
\item for every~$v \in V(G),$ the subgraph of~${T}$ induced by $\{t \in V(T)\mid {v \in X_t}\}$ is connected.
\end{itemize}
The {\em width} of $(T,\mathcal{X})$ is equal to $\max\big\{\left|X_t\right|-1 \mid t\in V(T)\big\}$ and the {\em treewidth} of $G$ is the minimum width over all tree decompositions of $G.$

The following result from \cite[Lemma 4.2]{GolovachKMT17} states that given a $q\in \mathbb{N}$ and a graph $G$  with treewidth more than $9q$, we can find a $q$-wall of $G$.

\begin{proposition}\label{prop:wall-vs-tw}
There exists an algorithm that receives as an input a planar graph $G$ and a $q\in \mathbb{N}$ and outputs, in $2^{q^{\mathcal{O}(1)}}\cdot |G|$ time, either a $q$-wall $W$ of $G$
or a tree decomposition of $G$ of width at most $9 q$.
\end{proposition}

\subsection{Frameworks}\label{subsec:frameworks}
We recall definitions related to frameworks.

\paragraph{Matroids.} We refer to the textbook of Oxley~\cite{Oxley11} for the introduction to Matroid Theory.

\begin{definition}\label{def:matroid}
A pair $M=(V,\mathcal{I})$, where $V$ is a \emph{ground set} and $\mathcal{I}$ is a family of subsets of $V$, called \emph{independent sets of $M$}, is a \emph{matroid} if it satisfies the following conditions, called \emph{independence axioms}:
\begin{itemize}
\item[~{\em (I1)}]  $\emptyset\in \mathcal{I}$, 
\item[~{\em (I2)}]  if $X \subseteq Y $ and $Y\in \mathcal{I}$ then $X\in\mathcal{I}$, 
\item[~{\em (I3)}] if $X,Y  \in \mathcal{I}$  and $ |X| < |Y| $, then there is $v\in  Y \setminus X $  such that $X\cup\{v\} \in \mathcal{I}$.
\end{itemize}
\end{definition}
An inclusion maximal set of $\mathcal{I}$ is called a \emph{base}.
We use $V(M)$ and $\mathcal{I}(M)$ to denote the ground set and the family of independent sets of $M$, respectively.  

Let $M=(V,\mathcal{I})$ be a matroid. We use $2^V$ to denote the set of all subsets of $V$.
A function $r\colon 2^V\rightarrow \mathbb{Z}_{\geq 0}$  such that for every $X\subseteq V$,
\begin{equation*}
r(X)=\max\{|Y|\colon Y\subseteq X\text{ and }Y\in \mathcal{I}\}
\end{equation*}
is called the \emph{rank function} of $M$. The \emph{rank of $M$}, denoted $r(M)$, is $r(V)$; equivalently, the rank of $M$ is the size of any base of $M$.  

\paragraph{Matroid representations.}
Let  $M=(V,\mathcal{I})$ be a matroid and let $\mathbb{F}$ be a field. An $r\times n$-matrix $A$ is a \emph{representation of $M$ over $\mathbb{F}$} if there is a bijective correspondence $f$ between $V$ and the set of columns of $A$ such that for every $X\subseteq V$, $X\in \mathcal{I}$ if and only if the set of columns $f(X)$ consists of linearly independent vectors of $\mathbb{F}^r$. Equivalently, $A$ is a representation of $M$ if $M$ is isomorphic to the \emph{column} matroid of $A$, that is, the matroid whose ground set is the set of columns of the matrix and the independence of a set of columns is defined as the linear independence.   
If $M$ has a  
such a representation, then $M$ is \emph{representable} over $\mathbb{F}$ and it is also said
$M$ is a \emph{linear} (or \emph{$\mathbb{F}$-linear}) matroid. 
We can assume that the number of rows $r=r(M)$ for a matrix representing $M$~\cite{Marx09}. 

Whenever we consider a linear matroid, it is assumed that its representation is given and the size of $M$ is $\|M\|=\|A\|$, that is, the bit-length of the representation matrix.
Notice that given a representation of a matroid, deciding whether a set is independent demands a polynomial number of field operations.
In particular, if the considered field is a finite or is the field of rationals, we can verify independence in time that is a polynomial in $\|M\|$. 
Another standard way to encode a matroid in problem inputs is by using \emph{independence oracles}. Such an oracle, given a subset of the ground set, in unit time correctly returns either \textsf{yes} or \textsf{no} depending on whether the set is independent or not.   Thus a matroid can be fully described by its ground set and the independence oracle.\medskip

\paragraph{Frameworks.}
A framework is a pair $(G, M)$, where $M = (V, \mathcal{I})$ is a matroid whose ground set is the set of vertices of $G$, i.e., $V(M) = V(G)$.
An $(s,t)$-path $P$ in a framework $(G,M)$ has  \emph{rank at least $k$} if there is a set $X \subseteq V(P)$ with $X \in \mathcal{I}$ and $|X| = k$.

 \section{Rerouting paths and cycles}\label{sec:det}

In this section, our goal is to prove~\autoref{thm:fpt-planar} that we restate here.

\thmirrelevantv*

The algorithm of~\autoref{thm:fpt-planar}
consists of two parts.
In the first part, we use the irrelevant vertex technique in order to design an algorithm that removes vertices form the input graph as long as its treewidth is big enough.
In order to do this, in \Cref{subsec:rerouting} we prove a combinatorial result (\autoref{lem:rerout}) that allows us to argue that,
given a planar graph and a wall of it and a vertex set $S$ that lies outside the wall, if there is a path $P$ that contains $S$ and  invades deeply enough inside the wall, we
can find another path $P'$ that contains $S$ (with the same endpoints as $P$) and avoids some ``central area'' of the wall.
Then, in \Cref{subsec:red-tw}, we give an algorithm (\autoref{lem:reducing-tw}) that given a planar graph of ``big enough'' (as a
function of $k$) treewidth, outputs, in time $2^{2^{\mathcal{O}(k\log k)}}\cdot (|G|+\|M\|)^{\mathcal{O}(1)}$, either a path of rank at least $k$ or an irrelevant vertex.
Finally, in \Cref{subsec:dp-planar}, we provide the dynamic programming algorithm that solves the problem in graphs of bounded treewidth.

\subsection{Rerouting paths and cycles}
 \label{subsec:rerouting}

In this subsection, we aim to prove the main combinatorial result (\autoref{lem:rerout}) that allows us to find an $(s,t)$-path that contains a given set $S$ and avoids some inner part of a given wall.
Before stating \autoref{lem:rerout}, we first prove the following result (\autoref{lem:linkage-kawa}) that will be an important tool for the proof of~\autoref{lem:rerout}.
The proof of~\autoref{lem:linkage-kawa} is inspired by the proof of \cite[Lemma~1]{Kawarabayashi08}.

\begin{lemma}\label{lem:linkage-kawa}
Let $G$ be a planar graph, let $k\in\mathbb{N}$, let $W$ be a wall of height at least $2k+3$.
Also, let $E=\{e_1, \ldots, e_k,e_{k+1},e_{k+2}\}$ be a set of $k+2$ edges of $G$, where, for every $i\in\{1,\ldots,k\}$, $e_i = \{v_i, u_i\}$, $e_{k+1} = \{v_{k+1}, s\}$, $e_{k+2}=\{v_{k+2},t\}$, and let $X$ be the set $\{v_{k+1}, v_{k+2}\}\cup \bigcup_{i\in\{1,\ldots,k\}}\{v_i,u_i\}$.
If every $v\in X$ is an in-peg of $\mathsf{perim}(W)$,
then there is an $(s,t)$-path in $G$ that contains the edges $e_1,\ldots, e_{k+2}$ and its intersection with $\cps(W^{(k+1)})$ is a path of $\mathsf{perim}(W^{(k+1)})$ whose endpoints are branch vertices of $W$.
\end{lemma}

\begin{proof}
Let $H$ be the graph whose vertex set is $\{s,t\}\cup X$ and whose edge set is $\{e_1, \ldots, e_{k+2}\}$.
Observe that $H$ is the disjoint union of $k+2$ edges.

We will prove the statement by induction on $k$.
If $k = 0$, then $|X|=2$, $H$ contains exactly two edges, $e_{1} = \{v_1, s\}$ and $e_{2} = \{v_2,t\}$.
By connecting $s$ and $t$ through a $(v_1,v_2)$-path in $\mathsf{perim}(W)$, we obtained the claimed $(s,t)$-path.

 \begin{figure}[h]
\centering
\includegraphics[width =7.5cm]{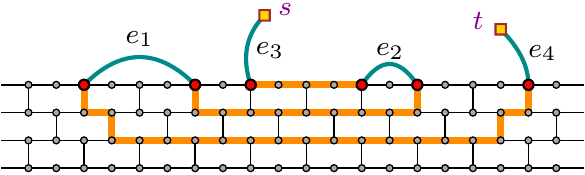}
\caption{Visualization of the proof of~\autoref{lem:linkage-kawa} for $k=2$.
In this example, the edges $e_1, \ldots, e_4$ are depicted in blue and the vertex set $X$ is depicted in red.
The highlighted orange paths inside the wall correspond to the paths used in the proof to construct the claimed $(s,t)$-path.}
\label{fig:route_wall}
\end{figure} 

Suppose that $k\geq 1$.
Take a vertex $x\in X$, let $e_x$ be the edge of $H$ that is incident to $x$,
and let $x'\in V(H)$ be the other endpoint of $e_x$.
Also, let $y$ be a vertex of $X\setminus \{x'\}$ such that there is an $(x,y)$-path $Q_{x,y}$ in $\mathsf{perim}(W)$ such that no internal vertex of $Q_{x,y}$ is in $X$
and let $e_y$ be the edge of $H$ that is incident to $y$.
By the choice of $y$ (i.e., $y\in X\setminus\{x\}$),
$e_x\neq e_y$.
We set $E'$ to be the edge set obtained after contracting the edges in $E(Q_{x,y})\cup\{e_x,e_y\}$ to a single edge, denoted by $e_x'$, and after contracting for each $v\in X\setminus \{x,y\}$, the (unique) edge in $E(\mathsf{inn}(W))$ that is incident to $v$.
Observe that $|E'| = k+1$ and the wall $W^{(2)}$ has height at least $2k+1$ and therefore we can apply the induction hypothesis to find the claimed $(s,t)$-path.
For an illustration of the obtained $(s,t)$-path, see~\autoref{fig:route_wall}.
This completes the proof of the lemma. 
\end{proof}

We are now ready to prove the following.

\begin{lemma}\label{lem:rerout}
There is a function $h:\mathbb{N} \to \mathbb{N}$ such that if $k,z\in\mathbb{N}$, $G$ is a planar graph, $s,t\in V(G)$,
$S$ is a subset of $V(G)$ of size at most $k$,
$W$ is a wall of $G$ of at least $h(k)$ layers and whose compass is disjoint from $S\cup\{s,t\}$,
and $P$ is an $(s,t)$-path of $G$ such that $S\subseteq V(P)$ and $P$
intersects $V(\mathsf{inn}(W^{(h(k))}))$,
then there is an $(s,t)$-path $\tilde{P}$ of $G$ such that $S\subseteq V(\tilde{P})$
and its intersection with $\cps(W^{(h(k))})$ is a path of $\mathsf{perim}(W^{(h(k))})$ whose endpoints are branch vertices of $W$.
Moreover, $h(k) = \mathcal{O}(k^2)$.
\end{lemma}

\begin{proof}
We set $h(k) :=2k\cdot  (k+2) +2k+1$.
Let $W$ be a wall of at least $h(k)$ layers.
For $i\in\{1,\ldots,k+2\}$, we use $C_i$ to denote the layer $L_{2k\cdot (i-1) +1}$ of $W$.
Intuitively, we take $C_1$ to be the first layer of $W$ and for every $i\in\{2,\ldots, k+2\}$, we take $C_i$ to be the $2k$-th consecutive layer after $C_{i-1}$.
Also, we use $D_i$ to denote the vertex set of $\cps(W^{(2k\cdot(i-1)+1)})$.
Keep in mind that $C_i$ is the perimeter of $W^{(2k\cdot(i-1)+1)}$.
For every $i\in [k+2]$, we consider the collection $\mathcal{F}_i$ of paths of $G$ that are subpaths of $P$ that intersect $D_i$ only on their endpoints and that there is an onto function mapping each vertex $u\in S\cup\{s,t\}$ to the path in $\mathcal{F}_i$ that contains $u$.
Intuitively, for each $u\in S\cup\{s,t\}$ we consider the maximal subpath of $P$ that contains $u$ and intersects $D_i$ only on its endpoints and we define $\mathcal{F}_i$ to be the collection of these maximal paths (see~\autoref{fig:contr_edg} for an example).

 \begin{figure}[h]
\centering
\includegraphics[width = 9cm]{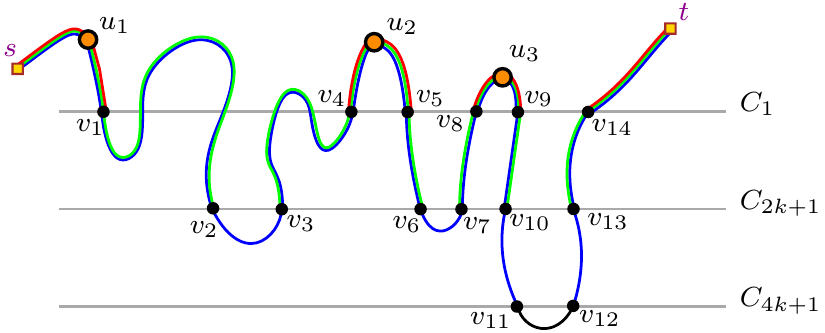}
\caption{An example of an $(s,t)$-path $P$ containing an independent set $S= \{u_1, u_2, u_3\}$.
In this example, $\mathcal{F}_1$ is the collection of the four red paths (the ones with endpoints $(s, v_1)$, $(v_4, v_5)$, $(v_8, v_9)$, and $(v_{14}, t)$), 
$\mathcal{F}_2$ is the collection of the four green paths (the ones with endpoints
$(s, v_2)$, $(v_3, v_6)$, $(v_7,v_{10})$, $(v_{13}, t)$), and
$\mathcal{F}_3$ is the collection of the two blue paths (the $(s,v_{11})$-path and the $(v_{12},t)$-path).}
\label{fig:contr_edg}
\end{figure} 

Observe that $|\mathcal{F}_1|\leq k+2$ (since $|S\cup\{s,t\}|\leq k+2$) and $|\mathcal{F}_{k+2}|\geq 2$ (since  $V(P)\cap V(\mathsf{inn}(W^{(h(k))}))\neq \emptyset$ and therefore $P$ intersects at least twice every $C_i$, $i\in [k+2]$).
For every  $i\in\{1,\ldots, k+2\}$, we assume that $\mathcal{F}_i = \{F_{i,1}, \ldots, F_{i,|\mathcal{F}_i|}\}$, where the ordering is given by traversing $P$ from $s$ to $t$.
For every $i\in\{1,\ldots,k+2\}$, we set $\mathcal{Q}_i = \{Q_{i,1},\ldots, Q_{i, |\mathcal{F}_i|-1}\}$, where, for each $j\in [|\mathcal{F}_i|-1]$, $Q_{i,j}$ is the minimal subpath of $P$ that intersects both $V(F_{i,j})$ and $V(F_{i,j+1})$.
Observe that, for every $i\in\{1,\ldots, k+2\}$, $P$ is the concatenation of the paths $F_{i,1}, Q_{i,1},F_{i,2}, \ldots, Q_{i,|\mathcal{F}_i|-1},F_{i,|\mathcal{F}_i|}$.
In~\autoref{fig:contr_edg}, $\mathcal{Q}_1 = \{Q_{1,1}, Q_{1,2}, Q_{1,3}\}$, where $Q_{1,1}$ is the $(v_1, v_4)$-subpath, $Q_{1,2}$ is the $(v_5, v_8)$-subpath, and $Q_{1,3}$ is the $(v_9, v_{14})$-subpath of $P$,
$\mathcal{Q}_2 = \{Q_{2,1}, Q_{2,2}, Q_{2,3}\}$,
where $Q_{2,1}$ is the $(v_2, v_3)$-subpath, $Q_{2,2}$ is the $(v_6, v_7)$-subpath,
and $Q_{2,3}$ is the $(v_{10}, v_{13})$-subpath of $P$, and $\mathcal{Q}_{3}$ consists of the $(v_{11},v_{12})$-subpath $Q_{3,1}$ of $P$.

It is easy to see that for every $i\in\{1,\ldots, k+1\}$, $|\mathcal{F}_{i+1}|$ is equal to $|\mathcal{F}_{i}|$ minus the number of paths in $\mathcal{Q}_i$ that do not intersect $C_{i+1}$ and therefore, $|\mathcal{F}_{i}|\leq |\mathcal{F}_{i+1}|$.
Therefore, given that $|\mathcal{F}_1|\leq k+2,$ $|\mathcal{F}_{k+2}|\geq 2$, and for every $i\in\{1,\ldots,k+1\}$, $|\mathcal{F}_{i}|\leq |\mathcal{F}_{i+1}|$,
there is an $i_0\in\{1,\ldots,k+1\}$ such that $|\mathcal{F}_{i_0}| = |\mathcal{F}_{i_0+1}|$ (if there are many such $i_0$, we pick the minimal one).
This implies that every path in $\mathcal{Q}_{i_0}$ intersects $C_{i_0 +1}$.

For each $F\in \mathcal{F}_{i_0}$, we denote by $v_F$ and $u_F$ the endpoints of $F$.
We define the graph $G'$ obtained from $G$ after removing the internal vertices of every $F\in \mathcal{F}_{i_0}$ (i.e., the vertex set $\bigcup_{F\in \mathcal{F}_{i_0}} (V(F)\setminus \{v_F,u_F\})$) and adding the edge $\{v_F,u_F\}$ for every $F\in \mathcal{F}_{i_0}$.
Observe that $G'$ is also planar and contains $D_{i_0}$ as a subgraph.
Moreover, notice that, for every $F\in \mathcal{F}_{i_0}$, $\{v_F,u_F\}\in V(C_{j_0})\cup\{s,t\}$.
In~\autoref{fig:contr_edg}, $|\mathcal{F}_1| = |\mathcal{F}_2|$ and thus $G'$ is obtained after replacing each 3-colored path with an edge.

In the rest of the proof we will argue that, in $G'$, there is an $(s,t)$-path that contains all edges $\{v_F,u_F\}$, $F\in \mathcal{F}_{i_0}$, and its intersection with $V(\cps(W^{(h(k))}))$ is the vertex set of a subdivided edge of $W$ that lies in $\mathsf{perim}(W^{(h(k))})$.
Having such a path in hand, we can replace each edge $\{v_F,u_F\}$, $F\in \mathcal{F}_{i_0}$ with the corresponding path $F$ and thus obtain the path $\tilde{P}$ claimed in the statement of the lemma.

We will denote by $C$ the cycle $C_{j_0}$ (that is the layer $L_{2k\cdot (i_0-1)+1}$) and by $C'$ the layer $L_{2k\cdot i_0}$.
To get some intuition, recall that $C_{i_0+1} = L_{2k\cdot i_0 +1}$ and therefore $C'$ is the layer of $W$ ``preceding'' $C_{i_0+1}$.
Since every path in $\mathcal{Q}_{i_0}$ intersects $C_{i_0+1}$, it holds that every path in $\mathcal{Q}_{i_0}$ intersects $C'$ at least twice.
Therefore, if we set $Y:=V(C)\cap \bigcup_{F\in \mathcal{F}_{i_0}} \{v_F,u_F\}$ and $ \ell:=|Y|$, then $\ell\leq 2k$ and there are $\ell$ disjoint paths from $Y$ to $C'$ (for an example, see the left part of~\autoref{fig:pathtobranch}).

Recall that $\mathsf{perim}(W^{(2k\cdot i_0)})=C'$.
We set $B$ to be the set of branch vertices of $W$ that are in $V(C')$ and have degree three in $W^{(2k\cdot i_0)}$.
Also, we set $\mathcal{K}$ to be the graph $G'\setminus V(\mathsf{inn}(W^{(2k\cdot i_0)}))$.
We now argue that there also exist $\ell$ disjoint paths from $Y$ to $B$ in $\mathcal{K}$.
 \begin{figure}[ht]
\centering
\includegraphics[width = 13cm]{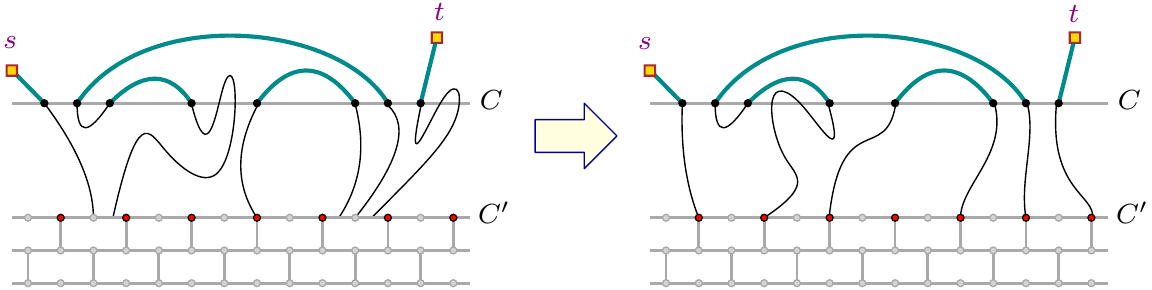}
\caption{A visualization of the statement of \autoref{cl:linkage-one}. In both figures, the edges $\{v_F,u_F\}$ are depicted in blue, the black vertices correspond to the set $Y$ and the red vertices correspond to the set $B$. In the left figure, we illustrate $|Y|$ disjoint paths from $Y$ to $C'$, while in the right figure, we illustrate $|Y|$ disjoint paths from $Y$ to $B$.}
\label{fig:pathtobranch}
\end{figure} 

\begin{claim}\label{cl:linkage-one}
There is a set $X\subseteq B$, a bijection $\rho: Y\to X$, and a collection $\mathcal{P} = \{P_v \mid v\in Y\}$ of pairwise disjoint paths where, for every $v\in Y$, $P_v$ is a $(v,\rho(v))$-path in $\mathcal{K}$.
\end{claim}

\begin{proof}[Proof of \autoref{cl:linkage-one}]
Suppose, towards a contradiction, that there are is a set $S\subseteq V(\mathcal{K})$ of size
at most $\ell-1$ such that there is no
path in $\mathcal{K}\setminus S$ from $Y$ to $B$.

Since there are  $\ell$ disjoint paths from $Y$ to $V(C')$,
there is a connected component $A$ of $\mathcal{K}\setminus S$ that
contains vertices from both $Y$ and $V(C')$.
Since $Y\subseteq V(C)$, $A$ contains vertices from both $V(C)$ and $V(C')$.
Also,
since $C' = L_{2k\cdot i_0}$ where $i_0\in\{1,\ldots,k+2\}$ and $W$ has at least $h(k)$ layers, where $h(k) > 2k\cdot  (k+2) +2k,$ there exist at least $2k$ vertex disjoint paths from $V(C)$ to $B$.
This, together with the fact that $|S|< \ell$ and $\ell \leq 2k$, implies that 
 there is a connected component $A'$ of $\mathcal{K}\setminus S$ that contains
vertices from both $V(C)$ and $B$.

Since both $A$ and $A'$ contain vertices of both $C$ and $C'$,
there exist paths $P,P'$ in $A$ and $A'$ respectively, both intersecting $V(C)$ and $V(C')$.
The fact that $C = L_{2k\cdot (i_0-1)+1}$ and $C'=L_{2k\cdot i_0}$
implies that there are $2k$ layers intersecting both $V(P)$ and $V(P')$,
that yield $2k$ disjoint paths between $V(P)$ and $V(P')$ in $\mathcal{K}$.
Since $S< \ell$ and $\ell\leq 2k$, some of the aforementioned disjoint paths between
 $V(P)$ and $V(P')$ should remain intact in $\mathcal{K}\setminus S$, implying that $A=A'$.
But, given that $A$ contains vertices from $Y$,
and $A'$ contains vertices from $B$, we conclude that $S$ does not separate $Y$ and $B$, a contradiction to the initial assumption.
Therefore, there exist $\ell$ disjoint paths from $Y$ to $B$.
We set $X$ to be the endpoints (in $B$) of these paths and this proves the claim.
\end{proof}

Following \autoref{cl:linkage-one}, let $X\subseteq B$, let a bijection $\rho: Y\to X$, and let a collection $\mathcal{P} = \{P_v \mid v\in Y\}$ of pairwise disjoint paths
such that for every $v\in Y$, $P_v$ is a $(v,\rho(v))$-path in $\mathcal{K}$.

Now, for each $F\in \mathcal{F}_{i_0}$, we consider the path $P_F$ obtained after joining the paths $P_{v_{F}}$ and $P_{u_F}$ by the edge $\{v_F, u_F\}$ (in the case where $s,t\in \{v_F,u_F\}$, we just extend the corresponding path in $\mathcal{P}$ by adding the edge $\{v_F,u_F\}$).
Let $G''$ be the graph obtained from $G'$ after contracting each $P_F, F\in \mathcal{F}_{i_0}$ to an edge $e_{P_F}$ and let $E = \{e_{P_F}\mid  F\in \mathcal{F}_{i_0}\}$.
Then, notice that $G''$ contains $W^{(2k\cdot i_0)}$ as a subgraph and since $h(k) = 2k\cdot (k+2)+2k+1$, the wall $W^{(2k\cdot i_0)}$ has at least $k+1$ layers and therefore height at least $2k+3$.
Therefore, by~\autoref{lem:linkage-kawa}, $G''$ contains an $(s,t)$-path that contains all edges in $E$ and its intersection with $\cps(W^{(2k\cdot i_0+k+1)})$ is a path of $\mathsf{perim}(W^{(2k\cdot i_0+k+1)})$ whose endpoints are branch vertices of $W$.

Thus, using this $(s,t)$-path in $G''$, we can find an $(s,t)$-path $P^\star$ in $G$ that contains $S$ and
its intersection with $\cps(W^{(2k\cdot i_0+k+1)})$ is a path of $\mathsf{perim}(W^{(2k\cdot i_0+k+1)})$ whose endpoints, say $x$ and $y$, are branch vertices of $W$.
Finally, let an $(x,y)$ path $R_{x,y}$ in $\cps(W^{(2k\cdot i_0+k+1)})$ whose intersection with 
$\cps(W^{(h(k))})$ is a path of $\mathsf{perim}(W^{(h(k))})$ whose endpoints are branch vertices of $W$.
The proof concludes by observing that $(P^\star\setminus V(\cps(W^{(2k\cdot i_0+k+1)})))\cup R_{x,y}$ is the $(s,t)$-path claimed in the statement of the lemma.
\end{proof}

We stress that, while~\autoref{lem:rerout} deals with the case of ``rerouting'' an $(s,t)$-path, we can apply the same arguments to ``reroute'' a cycle that contains a fixed set $S$ away from the inner part of some wall.

\subsection{Equivalent instances of small treewidth}
\label{subsec:red-tw}
In this subsection, we prove that there is an algorithm that receives a framework $(G,M)$, where $G$ is a planar graph of ``big enough'' treewidth, and two vertices $s,t\in V(G)$,
and outputs either a report that $G$ contains an $(s,t)$-path of rank at least $k$, or an irrelevant vertex that can be safely removed.
In frameworks, to remove a vertex, one has to remove this vertex from $G$ and also restrict the matroid.
 
 \paragraph*{Restrictions of matroids.}
 Let $M =(V, \mathcal{I})$ be a matroid and let $S\subseteq V$.
 We define the {\em restriction of $M$ to $S$}, denoted by $M|S$, to be the matroid on the set $S$ whose independent sets are the sets in $\mathcal{I}$ that are subsets of $S$.
Given a $v\in V$,
we denote by $M\setminus v$ the matroid $M|(V\setminus \{v\})$.
\medskip

The goal of this subsection is to prove the following.

\begin{lemma}\label{lem:reducing-tw}
There is a function $g:\mathbb{N}\to\mathbb{N}$ and an algorithm that, given an integer $k\in\mathbb{N}$, a framework $(G,M)$, where $M$ is a matroid for which we can verify independence in time $\|M\|^{\mathcal{O}(1)}$,
 and $G$ is a planar graph of treewidth at least $g(k)$, and two vertices $s,t\in V(G)$, outputs, in time $2^{2^{\mathcal{O}(k\log k)}}\cdot (|G|+\|M\|)^{\mathcal{O}(1)}$,
\begin{itemize}
\item either a report that $G$ contains an $(s,t)$-path of rank at least $k$, or
\item a vertex $v\in V(G)$ such that $(G,M,k,s,t)$ and $(G\setminus v, M\setminus v,k,s,t)$ are equivalent instances of \textsc{Maximum Rank  $(s,t)$-Path}.
\end{itemize}
Moreover, $g(k) = 2^{\mathcal{O}(k\log k)}$.
\end{lemma}
Keep in mind that, if $M$ is represented over a finite field or $\mathbb{Q}$, we can verify independence in time that is a polynomial in $\|M\|$. 
In order to prove \autoref{lem:reducing-tw}, we need some additional definitions and results.

\paragraph*{Packings of walls.}
Let $G$ be a planar graph and $W$ be a wall of $G$.
Let $z,q$ be two non-negative odd integers and let $r\in\mathbb{N}$.
We say that $W$ {\em admits an $(z,r,q)$-packing of walls},
if there is a collection $\mathcal{W} =\{W_0,W_1,\ldots, W_r\}$ of subwalls of $W$, where 
$W_0$ is a subwall of $W$ of height $j$, for some odd $j\geq 2z$, and
for every $i\in\{1,\ldots,r\}$, $W_i$ is a subwall of $W_0$ of height at least $q$ such that $V(W_i)$ is a subset of $V(W_0^{(z+1)})$, and
for every $i,j\in\{1,\ldots,r\},$ with $i\neq j$, $V(\cps(W_i '))$ and $V(\cps(W_j '))$ are disjoint.
We call $\mathcal{W}$ an {\em $(z,r,q)$-packing} of $W$ (see~\autoref{fig:reroute} for a visualization of a packing of a wall $W$).

\begin{observation}\label{obs:pack}
Given odd integers $z,q\in \mathbb{N}$, an $r\in\mathbb{N}$, and a planar graph $G$, every wall $W$ of $G$ of height at least $z+\lceil \sqrt{r}\rceil\cdot (q+1)$ admits a $(z,r,q)$-packing.
\end{observation}

Let $W$ be a wall of a planar graph.
We use $\rho(W)$ to denote $r(V(\cps(W)))$.

\begin{lemma}\label{lem:packings}
There is a function $f:\mathbb{N}^4\to\mathbb{N}$ and an algorithm that, given integers $k,z,r,q\in\mathbb{N}$, where $z,q$ are odd, a framework $(G,M)$, where $G$ is planar and $M$ is a matroid for which we can verify independence in time $\|M\|^{\mathcal{O}(1)}$, and a wall $W$ of $G$ of height at least $f(k,z,r,q)$ such that $\rho(W)\leq k$,
outputs, in $k\cdot r\cdot (|G|+\|M\|)^{\mathcal{O}(1)}$ time, a $(z,r,q)$-packing $\mathcal{W}=\{W_0,W_1,\ldots, W_r\}$ of $W$ such that
for every $i\in\{1,\ldots,r\}$, $\rho(W_i) = \rho(W_0)$.
Moreover, $f(k,z,r,q) = \mathcal{O}(r^{k/2}\cdot z\cdot q)$.
\end{lemma}

\begin{proof}
We define the function $f:\mathbb{N}^4\to\mathbb{N}$ so that, for every $z,r,q\in\mathbb{N}$, $f(1,z,r,q) = z+\lceil \sqrt{r}\rceil\cdot (q+1)$, while for $k>1$, we set $f(k,z,r,q) = z+\lceil \sqrt{r}\rceil\cdot (f(k-1,z,r,q)+1).$
Observe that, since $r,q$ are odd, $f(k,z,r,q)$ is odd for every $k,z\in\mathbb{N}$.

We prove the lemma by induction on $k$.
Clearly, if $k=1$, then the lemma holds trivially, as, by~\autoref{obs:pack}, there is a $(z,r,q)$-packing $\mathcal{W}$ of $W$ and also, given that for each $W'\in \mathcal{W}$ $\cps(W')$ is a subgraph of $\cps(W)$,  we have that $\rho(W')\leq \rho(W)$ and thus $\rho(W')=1$.

Suppose now that $k>1$ and that the lemma holds for smaller values of $k$.
We set $w=f(k-1,z,r,q)$.
Since $W$ has height at least $z+\lceil \sqrt{r}\rceil\cdot (w+1)$,~\autoref{obs:pack} implies that $W$ admits a $(z,r,w)$-packing $\mathcal{W} = \{W_0,W_1,\ldots, W_r\}$.
Since, by definition, for every $i\in\{1,\ldots,r\}$ $V(\cps(W_i))$ is a subset of $V(\cps(W_0))$,
it holds that $\rho(W_i)\leq \rho(W_0)$, for every $i\in\{1,\ldots,r\}$.
We compute $\rho(W_0)$ and $\rho(W_i)$, for every $i\in\{1,\ldots,r\}$.
This can be done in $r\cdot (|G|+\|M\|)^{\mathcal{O}(1)}$ time.
If there is an $i\in\{1,\ldots,r\}$ such that $\rho(W_i)<\rho(W_0)$, then, from the induction hypothesis applied to $W_i$, we have that there exists a $(z,r,q)$-packing $\mathcal{W}_i$ of $W_i$ such that all walls in $\mathcal{W}'$ have the same rank.
The lemma follows by observing that $\mathcal{W}_i$ is also a $(z,r,q)$-packing of $W$ and that $f(k,z,r,q)=\mathcal{O}(r^{k/2}\cdot z\cdot q).$
\end{proof}

We are now ready to prove \autoref{lem:reducing-tw}.

\begin{proof}[Proof of \autoref{lem:reducing-tw}]
We set
\begin{eqnarray*}
b\!\!\! & =&\!\!\!\! h(k), \hspace{3.63cm} x  =  k+1, \hspace{4.65cm} z  = (k+1)\cdot b,\\
 q \!\!\! & =&\!\!\!\!  f(k-1,z,x,3), \hspace{1.9cm}
r  = 1+\lceil\sqrt{k}\rceil\cdot (q+1),\hspace{2cm}
g(k) = 36(r+1).
\end{eqnarray*}
We first assume that $G$ is 2-connected.
If $G$ is not connected, then we break the problem in subproblems, each one corresponding to a 2-connected component $B$ of $G$ and if the vertices of $B$ are separated from $s$ or $t$ by a cut-vertex $v$ of $G$, then we consider the problem where $v$ is set to be $s$ or $t$, respectively.

Since the treewidth of $G$ is at least $g(k) = 36(r+1)$, by~\autoref{prop:wall-vs-tw}, there is a $(4r+1)$-wall of $G$.
We then consider an $r$-wall $W$ of $G$ such that $s,t\notin\cps(W)$ and  an $(1,k,q)$-packing $\mathcal{W} = \{W_0,W_1,\ldots, W_k\}$ of $W$.
This $(1,k,q)$-packing exists because of the fact that $r= 1+\lceil\sqrt{k}\rceil\cdot (q+1)$ and due to~\autoref{obs:pack} and we can find it in $\mathcal{O}(n)$ time.
For every $i\in\{1,\ldots,k\}$, we set $K_i := V(\cps(W_i))$.
Then, compute the rank of $K_i$, for each $i\in\{1,\ldots,k\}$.
This can be done in time $k\cdot (|G|+\|M\|)^{\mathcal{O}(1)}.$

If every $K_i$ has rank at least $k$, then notice that
there is a set $S\subseteq V(G)$ such that $r(S)=k$ and for every $i\in\{1,\ldots,k\}$, $|S\cap K_i| = 1$.
To obtain an $(s,t)$-path $P$ such that $S\subseteq V(P)$, we do the following:
We first pick two disjoint paths $P_s, P_t$ from the perimeter of $W_0$ to $s$ and $t$ respectively (these exist since $G$ is 2-connected).
Let $D$ be the perimeter of $W_0$ and let $s'$ and $t'$ be the endpoints of $P_s$ and $P_t$ in $D$.
Also, let $L_2$ be the second layer of $W_0$.
Observe that, since the compass of a wall is a connected graph, there is also a path $\overline{P}$ in $G$ such that the endpoints, say $x,y$, of $\overline{P}$ are in $L_2$, no internal vertex of $\overline{P}$ is a vertex of $L_2$, and $S\subseteq V(\overline{P})$.
Finally, observe that there exist two disjoint paths $P_{s'x}, P_{t'y}$ in the closed disk bounded by $D$ and $L_2$ connecting $s'$ with $x$ and $t'$ with $y$, respectively, and that $P:=P_s\cup P_{s'x}\cup \overline{P}\cup P_{t'y}\cup P_t$ is an $(s,t)$-path such that  $S\subseteq V(P)$ (see~\autoref{fig:path}).

Suppose now that there is an $i\in\{1,\ldots,k\}$ such that the rank of $K_i$ is at most $k-1$.
Since the corresponding wall $W_i$ has height at least $q = f(k-1,z,x,3)$, by \autoref{lem:packings}, we can find a $(z,x,3)$-packing $\mathcal{W}' = \{W_0', W_1',\ldots, W_x'\}$ of $W_i$, such that for every $j\in\{1,\ldots, x\}$, $\rho(W_j') = \rho(W_0')$.
We set $v$ to be one central vertex of $W_1'$.

We now prove that $(G,M,k,s,t)$ and $(G\setminus v, M\setminus v,k,s,t)$ are equivalent instances of \textsc{Maximum Rank  $(s,t)$-Path}.
We show that if $(G,M,k,s,t)$ is a \textsf{yes}-instance, then $(G\setminus v, M\setminus v,k,s,t)$ is also a \textsf{yes}-instance, since the other implication is trivial.
If $(G,M,k,s,t)$ is a \textsf{yes}-instance, then there is a set of vertices $S=\{v_1, \ldots, v_k\}\subseteq V(G)$ and an $(s,t)$-path $P$ in $G$ such that $r(S) = k$ and $S\subseteq V(P)$.
The fact that $z= (k+1)\cdot b$ implies that there is an $i\in\{1,\ldots, k+1\}$ such that the vertex set $V(\cps(W^{\prime((i-1)\cdot b+1)})\setminus V(\mathsf{inn}(W^{\prime(i\cdot b)})))$, which we denote by $D_{i}$, does not intersect $S$.
Let $S_\mathrm{in}$ be the vertices of $S$ that are contained in $\cps(W^{\prime ({i\cdot b})})$ and let $S_\mathrm{out}$ be the set $S\setminus S_\mathrm{in}$.
We will show that there is a set $S'\in \mathcal{I}(M\setminus v)$ and a path $P'$ such that $|S_\mathrm{out}\cup S'|\geq k$, $S_\mathrm{out}\cup S'\subseteq V(P')$ and $V(P')\subseteq V(G\setminus v)$.

We assume that $v\in V(P)$, otherwise we set $S':=S_\mathrm{in}$ and $P':=P$ and the lemma follows.
By \autoref{lem:rerout},
there is a path $\tilde{P}$ such that $S_\mathrm{out} \subseteq V(\tilde{P})$
and $V(\tilde{P}) \cap V(\cps(W_0^{\prime (i\cdot b)}))$ is the vertex set of a path  $\hat{P}$ of $W_0'$ that lies in $\mathsf{perim}(W_0^{\prime(i\cdot b)})$ and whose endpoints are branch vertices of $W_0^{\prime(i\cdot b)}$.
Let $v_{\hat{P}}$ and $u_{\hat{P}}$ be the endpoints  of  $\hat{P}$.

For every $j\in\{1,\ldots, x\}$, since $\rho(W_j ') = \rho(W_0 ')$ and $S_\mathrm{in}$ is an independent set of $M$ that is a subset of $\cps(W_0)$,
there is an independent set $S_j$ such that $|S_j| = |S_\mathrm{in}|$.
We set $S' = \{y_2, \ldots, y_{x}\}$, where $y_j$ is a vertex in $S_j$, $j\in\{1,\ldots, x\}$.
Observe that $S'$ is an independent set of $M$ of size $x-1$.
Since $x=k+1$, we have that $|S'| = k$ and therefore $|S_\mathrm{out}\cup S'|\geq k$.  
Also, notice that, for every $u_1,u_2\in L_{z}$, there is a $(u_1, u_2)$-path $P^\star$ in $W^{(z)}\setminus (V(L_z)\setminus \{u_1,u_2\})$ that contains $S'$ and avoids $v$.
It is easy to see that there exist two disjoint paths $Q_1,Q_2$ in $\cps(W_0^{\prime(i\cdot b)})$
connecting $\{v_{\hat{P}},u_{\hat{P}}\}$ with $\{u_1,u_2\}$ and that these paths can be picked to be internally disjoint from $\hat{P}$ and $P^\star$.
Thus, if $\tilde{P}''$ is the graph obtained from $\tilde{P}'$ after removing all internal vertices of $\hat{P}$, then $\tilde{P}''\cup Q_1 \cup Q_2 \cup P^\star$ is the claimed $(s,t)$-path that contains $S'\cup S_\mathrm{out}$ and avoids $v$ (see~\autoref{fig:reroute}).
\end{proof}

\section{Dynamic programming for instances of small treewidth}
\label{subsec:dp-planar}

In this section, we aim to describe a dynamic programming algorithm that solves 
\textsc{Maximum Rank  $(s,t)$-Path} for frameworks $(G,M)$, where $G$ has treewidth at most $q$ and $M$ is a linear matroid.
 
\begin{lemma}\label{lem:dp-planar}
Let $\mathbb{F}$ be a finite field or $\mathbb{Q}$.
There is an algorithm that, given a framework $(G,M)$, where $M$ is an $\mathbb{F}$-linear matroid and $G$ is a graph,
two non-negative integers $k$ and $q$, where $k\leq q$, and a tree decomposition of $G$ of width $q$, outputs, in time $2^{q^{\mathcal{O}(1)}}\cdot (|G|+\|M\|)^{\mathcal{O}(1)}$,
 a report whether $G$ contains an $(s,t)$-path of rank at least $k$ or not.
\end{lemma}

The section is organized as follows.
In~\Cref{subsec:nice-repsets}, we define nice tree decompositions, the combinatorial structure on which we will perform the dynamic programming, and representative sets, which are used to efficiently encode partial solutions to tables of the dynamic programming.
In~\Cref{subsec:partial-solutions}, we define the partial solutions of our problem.
Then, in~\Cref{subsec:dyn-prog}, we present the dynamic programming algorithm of~\Cref{lem:dp-planar} and in~\Cref{subsec:corr} we prove its correctness.
We conclude this section by giving the proof of~\Cref{thm:fpt-planar} (\Cref{subsec:proofThm2}).

\subsection{Nice tree decompositions and Representative Sets}\label{subsec:nice-repsets}
We start this subsection with the definition of nice tree decompositions.
\paragraph*{Nice tree decompositions.}
Let $G$ be a graph.
A tree decomposition $\mathcal{T}=(T,\mathcal{X})$ of $G$ is called {\em nice tree decomposition} of $G$ if $T$ is rooted to some leaf $r$ and
\begin{itemize}
\item for any leaf $l\in V(T)$, $X_l=\emptyset$ (we call $X_l$ {\em leaf node} of $\mathcal{T}$, except from $X_r$ which we call {\em root node}),
\item every $t\in V(T)$ has at most two children,
\item if $t$ has two children $t_1$ and $t_2$, then $X_t= X_{t_1}= X_{t_2}$ and $X_t$ is called a {\em join node},
\item if $t$ has one child $t'$, then
\begin{itemize}
\item either $X_{t} = X_{t'}\cup\{v\}$ for some $v\in V(G)$ (we call $X_t$ an {\em insert node}),
\item or $X_t = X_{t'}\setminus \{v\}$ for some $v\in V(G)$  (we call $X_t$ a {\em forget node}).
\end{itemize}
\end{itemize}

It is known that any tree decomposition of $G$ can be transformed into a nice tree decomposition maintaining the same width in linear time~\cite{Kloks94}.
We use $G_t$ to denote the graph induced by the vertex set $\bigcup_{t'} X_{t'}$,
where $t'$ ranges over all descendants of $t$, including $t$.

In the rest of the paper, given an instance $(G,M,k,s,t)$ of \textsc{Maximum Rank  $(s,t)$-Path} and a nice tree decomposition $(T, \mathcal{X})$ of $G$,
we will consider the tree decomposition $(T, \mathcal{X}')$
obtained from $(T, \mathcal{X})$ after adding 
the vertices
$s$ and $t$ in every bag in $\mathcal{X}$.
Therefore, the leaf nodes and the root node will be equal to the set $\{s,t\}$.
\medskip

Let $(G,M)$ be a framework and let $(T,\mathcal{X})$ be a tree decomposition of $G$.
For every $t\in V(T)$ and every $i\in\mathbb{N}$, we define $\mathcal{S}_t^{(i)}$ to be the collection of all sets $S\subseteq V(G_t)\setminus X_t$ that are independent sets of $M$ of size $i$.

\paragraph*{Representative sets.}  Our algorithms use results obtained by Fomin et al.~\cite{FominLPS16} and Lokshtanov et al.~\cite{LokshtanovMPS18}.
\begin{definition}[$q$-Representative Set]\label{def:repset}
Let $M=(V,\mathcal{I})$ be a matroid and let $\mathcal{S}$ be a family of subsets of $V$. For a positive integer $q$, a subfamily $\widehat{\mathcal{S}}$ is \emph{$q$-representative for $\mathcal{S}$}  if the following holds: for every set $Y\subseteq V$ of size at most $q$, if there us a set $X\in\mathcal{S}$ disjoint from $Y$ with $X\cup Y\in \mathcal{I}$, then there is $\widehat{X}\in\widehat{\mathcal{S}}$ disjoint from $Y$ with $\widehat{X}\cup Y\in \mathcal{I}$. 
 \end{definition}
We write $\widehat{\mathcal{S}}\subseteq_{rep}^q\mathcal{S}$ to denote that $\widehat{\mathcal{S}}\subseteq\mathcal{S}$ is  $q$-representative for $S$. It is crucial for us that representative families can be computed efficiently for linear matroids. To state these results, we say that a family of sets $\mathcal{S}$ is a \emph{$p$-family} for an integer $p\geq 0$ if $|S|=p$ for every $S\in\mathcal{S}$. 

\begin{theorem}[{\cite[Theorem 3.8]{FominLPS16}}]\label{thm:rep-rand}
Let  $M=(V,\mathcal{I})$ be a linear matroid and let $\mathcal{S}=\{S_1,\ldots,S_t\}$ be a $p$-family of independent sets. Then there exists    $\widehat{\mathcal{S}}\subseteq_{rep}^q\mathcal{S}$ of size at most $\binom{p+q}{p}$. Furthermore, given a representation $A$ of $M$ over a field $\mathbb{F}$, there is a randomized algorithm computing  $\widehat{\mathcal{S}}\subseteq_{rep}^q\mathcal{S}$ of size at most $\binom{p+q}{p}$ in $\mathcal{O}(\binom{p+q}{p}tp^\omega+t\binom{p+q}{q}^{\omega-1})+\|A\|^{\mathcal{O}(1)}$ operations over $\mathbb{F}$, where 
$\omega$ is the exponent of matrix multiplication.\footnote{The currently best value is $\omega\approx 2.3728596$~\cite{AlmanW21}.}
\end{theorem}

Observe that the algorithm in Theorem~\ref{thm:rep-rand} is randomized. This is due to fact that one of the steps of the algorithm is constructing of a $k$-truncation\footnote{A matroid
$M'=(V,\mathcal{I}') $ is a \emph{$k$-truncation} of $M=(V,\mathcal{I})$ if for every $X\subseteq V$, $X\in \mathcal{I}'$ if and only if $X\in \mathcal{I}$ and $|X|\leq k$.
} of $M$ for $k=p+q$.  
A $k$-truncation can be constructed algorithmically for linear matroids, but for general linear matroids, only a randomized algorithm is known~\cite{Marx09}. 
In~\cite{LokshtanovMPS18}, Lokshtanov et al. gave a deterministic algorithm for linear matroid represented over any
field in which the field operations can be done efficiently. In particular, this includes any finite field and the field of rational numbers. This way, they obtained the following theorem.

\begin{theorem}[{\cite[Theorem 1.3]{LokshtanovMPS18}}]\label{thm:rep-det}
Let  $M=(V,\mathcal{I})$ be a linear matroid of rank $r$ 
and let $\mathcal{S}=\{S_1,\ldots,S_t\}$ be a $p$-family of independent sets. 
Let $A$ be an $r\times |V|$-matrix representing $M$ over a field $\mathbb{F}$, and let $\omega$ is the exponent of matrix multiplication.
Then there are deterministic algorithms computing $\widehat{\mathcal{S}}\subseteq_{rep}^q\mathcal{S}$ as follows:
\begin{itemize}
\item A family $\widehat{\mathcal{S}}$ of size at most $\binom{p+q}{p}$ in $\mathcal{O}(\binom{p+q}{p}^2tp^3r^2+t\binom{p+q}{q}^\omega rp)+(r+|V|)^{\mathcal{O}(1)}$ operations over $\mathbb{F}$.
\item A family $\widehat{\mathcal{S}}$ of size at most $rp\binom{p+q}{p}$ in $\mathcal{O}(\binom{p+q}{p}tp^3r^2+t\binom{p+q}{q}^{\omega-1} (rp)^{\omega-1})+(r+|V|)^{\mathcal{O}(1)}$ operations over $\mathbb{F}$.
\end{itemize}
\end{theorem}

\subsection{Partial solutions}\label{subsec:partial-solutions}
We start by defining the notion of \emph{semi-mathchings}, that intuitively encode parts of a path.

\paragraph*{Semi-matchings.}
Let $X$ be a set.
Let $H$ be a graph whose vertex set is $X$, every vertex has degree at most two, and it is acyclic.
The collection $\mathcal{M}$ of the edges and the isolated vertices of $H$ is called a {\em semi-matching} of $X$.
Given a semi-matching $\mathcal{M}$ of a set $X$, we use
$U(\mathcal{M})$ to denote $X$.
Observe that $|\{\mathcal{M}\mid \mathcal{M} \text{ is a semi-matching of $X$}\}| = 2^{\mathcal{O}(|X|\log |X|)}$.
We denote by $\mathcal{M}^{(v)}_1$ the set $\{\{v\}\}\cap \mathcal{M}$, by $\mathcal{M}^{(v)}_2$ the set $\{\{u,v\}\mid \{u,v\}\in \mathcal{M}\}$, and by $\mathcal{M}^{(v)}$ the set $\mathcal{M}^{(v)}_1\cup \mathcal{M}^{(v)}_2$.
Notice that $|\mathcal{M}^{(v)}_2|\leq 2$.

Given a semi-matching $\mathcal{M}$ of a set $X$ and a $v\in X$, we denote by $\mathsf{rem}(\mathcal{M},v)$ the semi-matching $\mathcal{M}'=(\mathcal{M}\setminus \mathcal{M}^{(v)})\cup\{\{u\}\mid \{u,v\}\in \mathcal{M}^{(v)} \text{ and }u\not\in U(\mathcal{M}\setminus \mathcal{M}^{(v)})\}$.
Also, given a set $Y$ such that $X\subsetneq Y$ and a $u\in Y\setminus X$,
we denote by $\mathsf{add}(\mathcal{M},u)$ the collection of all semi-matchings $\mathcal{M}'$ of $X\cup\{u\}$ such that $\mathcal{M}= \mathsf{rem}(\mathcal{M}',v)$.

\paragraph*{Linear forests.}
We say that a graph $F$ is a {\em linear forest} if it is an acyclic graph of maximum degree two.
Let $G$ be a graph and let $F$ be a linear forest that is a subgraph of $G$.
Given a set $X\subseteq V(G)$, we define $\mathsf{sig}_{G,X}(F)$
to be the set that contains (i) all vertices in $X\cap V(F)$ that have degree zero in $F$ and (ii) all pairs $\{u,v\}$ of vertices in $X\cap V(F)$ such that either $\{u,v\}\in E(F)$ or there is a $(u,v)$-path in $F$ that intersects $X$ only at its endpoints.
Notice that $\mathsf{sig}_{G,X}(F)$ is a semi-matching of $X\cap V(F)$.
\bigskip

We are now ready to define what is considered as a partial solution to our problem.

\paragraph*{Partial solutions.}
Let $G$ be a graph, let $s,t\in V(G)$, and let $\mathcal{T} = (T,\mathcal{X})$ be a nice tree decomposition of $G$.
Given a $t\in V(T)$, we define a {\em partial solution} at $t$ to be a quadraple $(X,\mathcal{M}, i, S)$, where $\{s,t\}\subseteq X\subseteq X_t$, 
$\mathcal{M}$ is a semi-matching of $X$,
$i\in\{1,\ldots, k\}$,
and $S\in \mathcal{S}_t^{(i)}$,
such that there is a linear forest $F\subseteq G_t$ where
$X = V(F)\cap X_t$, $\mathcal{M} = \mathsf{sig}_{G_t, X_t}(F)$, and $S\subseteq V(F)$.
Keep in mind that $S\subseteq V(G_t \setminus X_t)$ and therefore $S\subseteq V(F\setminus X_t)$.
We also say that the linear forest $F$ {\em certifies} that $(X,\mathcal{M}, i, S)$ is a partial solution at $t$.
We denote by $\mathcal{B}_t$ the set of all partial solutions at $t$.
\medskip

We can easily observe the following.
\begin{observation}\label{obs:solut}
Let $(G,M)$ be a framework and $k\in \mathbb{N}$.
Then $G$ contains an $(s,t)$-path of rank at least $k$ if and only if there is a set $S\in \mathcal{S}_{r}^{(k)}$ such that $(\{s,t\},\{\{s,t\}\}, k, S)\in \mathcal{B}_r$.
\end{observation}

\subsection{A dynamic programming algorithm}\label{subsec:dyn-prog}

We are now ready to describe the dynamic programming algorithm of~\Cref{lem:dp-planar}.
For every $t\in V(T)$, we aim to construct a collection $\mathcal{F}_t\subseteq \mathcal{B}_t$ of partial solutions whose size is ``small''
since we cannot afford to store all independent sets of size $i$ and therefore all partial solutions in $\mathcal{B}_t$.
For this reason, we will use {\sl respresentative sets}, instead of all possible independent sets using~\autoref{thm:rep-det} and thus, for every $X\subseteq X_t$, every semi-matching $\mathcal{M}$ of $X$, and every $i\in\{1,\ldots, k\}$, we will keep only a ``representative'' collection of independent sets $\widehat{\mathcal{S}}\subseteq \mathcal{S}_{t}^{(i)}$ such that
for every $S\in \mathcal{S}_{t}^{(i)},$ $(X,\mathcal{M}, i, S)\in \mathcal{B}_t$ if and only if there is a $S'\in \widehat{\mathcal{S}}$ such that $(X,\mathcal{M}, i, S')\in \mathcal{F}_t$.
Given a $p$-family $\mathcal{S}$ of independent sets of a matroid $M$, we use $\mathsf{Rep}(\mathcal{S})$ to denote the $k$-representative subfamily $\widehat{\mathcal{S}}$ for $\mathcal{S}$ given by~\autoref{thm:rep-det}.

\paragraph*{Leaf node $t$.}
Here, as $X_t = \{s,t\}$, the graph $G_t\setminus X_t$ is empty and therefore we set $\mathcal{F}_t =\{(\{s,t\},\{\{s,t\}\},0,\emptyset)\}.$

\paragraph*{Insert node $t$ with child $t'$.}
We know that $X_t\supseteq X_{t'}$ and $|X_t| = |X_{t'}|+1$.
Let $v$ be the vertex in $X_t\setminus X_{t'}$.
For every $X\subseteq X_t$ that contains $s$ and $t$, every semi-matching $\mathcal{M}$ of $X$ and every $i\in\{0,\ldots, k\}$,
we set
\[
\mathcal{S}_t [X,\mathcal{M},i] =\begin{cases}
\{S\mid (X, \mathcal{M}, i,S)\in \mathcal{F}_{t'}\},  &  \text{if $v\notin X$,}\\
\{S\mid (X\setminus \{v\}, \mathsf{rem}(\mathcal{M}, v), i,S)\in \mathcal{F}_{t'}\},  & \text{if $v\in X$ and $\mathcal{M}^{(v)}_2\subseteq  E(G_t)$,}\\
\emptyset, & \text{if otherwise.}
\end{cases}
\]
We set $\mathcal{F}_t = \{(X,\mathcal{M}, i, S)\mid S\in \mathsf{Rep}(\mathcal{S}_t [X,\mathcal{M},i])\}$.

\paragraph*{Forget node $t$ with child $t'$.}
We know that $X_t\subseteq X_{t'}$ and $|X_t| = |X_{t'}|-1$.
Let $v$ be the vertex in $X_{t'}\setminus X_t$.
For every $X\subseteq X_t$  that contains $s$ and $t$, every semi-matching $\mathcal{M}$ of $X$ and every $i\in\{0,\ldots, k\}$,
we set
\begin{eqnarray*}
\mathcal{S}_t [X,\mathcal{M},i]\!\!\!\!& =\!\!\!\! &
\{S\mid (X, \mathcal{M}, i,S)\in \mathcal{F}_{t'}\}\\
& & \cup\ \big\{S\mid \exists \mathcal{M}'\in \mathsf{add}(\mathcal{M},v): (X\cup\{v\},\mathcal{M}', i,S)\in \mathcal{F}_{t'}\big\}\\
 & & \cup\ \big\{S\cup\{v\}\mid \exists \mathcal{M}'\in \mathsf{add}(\mathcal{M},v):(X\cup\{v\},\mathcal{M}', i-1, S)\in \mathcal{F}_{t'}\ \&\ S\cup\{v\}\in \mathcal{I}(M) \big\}
\end{eqnarray*}
We set $\mathcal{F}_t = \{(X,\mathcal{M}, i, S)\mid S\in \mathsf{Rep}(\mathcal{S}_t [X,\mathcal{M},i])\}$.

\paragraph*{Join node $t$ with children $t_1$ and $t_2$.}
We know that $X_t = X_{t_1} =X_{t_2}$.
Given a semi-matching $\mathcal{M}$ of a set $X$, we denote by $\xi(\mathcal{M})$ the set of all pairs $(\mathcal{M}_1,\mathcal{M}_2)$ such that $\mathcal{M}_1, \mathcal{M}_2\subseteq \mathcal{M}$, $\mathcal{M}_1 \cup\mathcal{M}_2 = \mathcal{M}$, and $\mathcal{M}_1\cap \mathcal{M}_2 = \emptyset$.
For every $X\subseteq X_t$  that contains $s$ and $t$, every semi-matching $\mathcal{M}$ of $X$ and every $i\in\{0,\ldots, k\}$,
we set
\begin{eqnarray*}
\mathcal{S}_t[X,\mathcal{M},i] =
\{ S_1 \cup S_2 & \mid  &\exists (\mathcal{M}_1,\mathcal{M}_2)\in \xi (\mathcal{M})\ \exists i_1, i_2\in \{0,\ldots, k\} : i_1+i_2= i\text{ and,}\\
& & \text{ if $X_i=U(\mathcal{M}_i),i\in\{1,2\}$, then $S_1\cup S_2\in \mathcal{I}(M)$,}\\
& & ~(X_1, \mathcal{M}_1, i_1, S_1)\in \mathcal{F}_{t_1}\text{ and }(X_2, \mathcal{M}_2, i_2, S_2)\in \mathcal{F}_{t_2}\}
\end{eqnarray*}
We set $\mathcal{F}_t = \{(X,\mathcal{M}, i, S)\mid S\in \mathsf{Rep}(\mathcal{S}_t [X,\mathcal{M},i])\}$.
\bigskip

Our dynamic programming algorithm computes $\mathcal{F}_t$ for every $t\in V(T)$ in a bottom-up manner and checks whether there is a set $S\in 2^{V(G)}\cap \mathcal{I}(M)$
of size $k$ such that
$(\{s,t\}, \{\{s,t\}\}, k, S)\in \mathcal{F}_r$.
If so, it outputs a report that there is an $(s,t)$-path of $G$ of rank at least $k$, otherwise it outputs a report that such a path does not exist.

\subsection{Proof of correctness of the dynamic programming algorithm}\label{subsec:corr}

To prove the correctness of the algorithm presented in~\Cref{subsec:dyn-prog},
we first prove the following.

\begin{lemma}\label{lem:dp-partial}
For every $t\in V(T)$, $\mathcal{F}_t\subseteq\{(X,\mathcal{M}, i, S)\mid S\in \mathcal{S}_t [X,\mathcal{M},i]\}\subseteq  \mathcal{B}_t$ and $|\mathcal{F}_t|=2^{q^{\mathcal{O}(1)}}$.
\end{lemma}

\begin{proof}
We prove the lemma by bottom-up induction on the decomposition tree.
Let $t\in V(T)$.
We distinguish cases depending on the type of node $X_t$.
\medskip

\noindent{\em Case 1:} $X_t$ is a leaf node.\medskip

In this case, the statement of the lemma holds trivially.
\medskip

In the following cases (i.e., when $X_t$ is either an insert node, a forget node, or a join node),
we will show that $\{(X,\mathcal{M}, i, S)\mid S\in \mathcal{S}_t [X,\mathcal{M},i]\}\subseteq \mathcal{B}_t$, since, due to~\autoref{thm:rep-det},
$\mathcal{F}_t\subseteq \{(X,\mathcal{M}, i, S)\mid S\in \mathcal{S}_t [X,\mathcal{M},i]\}$ and $|\mathcal{F}_t| =2^{q^{\mathcal{O}(1)}}$.
\medskip

\noindent{\em Case 2:} $X_t$ is an insert node.\medskip

Let $t'$ be the child of $t$.
Let $(X,\mathcal{M}, i, S)$ such that $S\in  \mathcal{S}_t [X,\mathcal{M},i]$.
If $v\notin X$, then $(X,\mathcal{M}, i, S)\in \mathcal{F}_{t'}$. By the induction hypothesis, there is a linear forest $F'$ that certifies that $(X,\mathcal{M}, i, S)\in \mathcal{B}_{t'}$.
Since $v\notin X$ and $X_t = X_{t'}\cup\{v\}$, $F'$ is also a linear forest in $G_t$ where $X= V(F')\cap X_t$ and $\mathcal{M}=\mathsf{sig}_{G_t, X_t}(F')$. Therefore, $F'$ certifies that $(X,\mathcal{M}, i, S)\in \mathcal{B}_t$.
If $v\in X$ and $\mathcal{M}^{(v)}_2\subseteq  E(G_t)$, then $(X\setminus \{v\}, \mathsf{rem}(\mathcal{M}, v), i,S)\in \mathcal{F}_{t'}$ and therefore, by the induction hypothesis, $(X\setminus \{v\}, \mathsf{rem}(\mathcal{M}, v), i,S)\in \mathcal{B}_{t'}$.
This implies that there is a linear forest $F'$ certifying that
$(X\setminus \{v\}, \mathsf{rem}(\mathcal{M}, v), i,S)\in \mathcal{B}_{t'}$.
Since $\mathsf{sig}_{G_{t'},X_{t'}}(F')=  \mathsf{rem}(\mathcal{M}, v)$, $\mathcal{M}$ is a semi-matching of $X$, and $\mathcal{M}^{(v)}_2\subseteq  E(G_t)$, we have that $F\cup\{v, \mathcal{M}^{(v)}_2\}$ is a linear forest, which we denote by $F'$.
Observe that $F'$ certifies that $(X,\mathcal{M}, i, S)\in \mathcal{B}_t$.
\medskip

\noindent{\em Case 3:} $X_t$ is a forget node.\medskip

Let $t'$ be the child of $t$ and let $(X,\mathcal{M}, i, S)$ such that $S\in \mathcal{S}_t [X,\mathcal{M},i]$.
Observe that, if $(X, \mathcal{M}, i, S)\in \mathcal{F}_{t'}$, there is a linear forest $F'$ certifying that $(X,\mathcal{M},i,S)\in \mathcal{B}_{t'}$ and the fact that $v\notin X_t$ (and therefore $v\notin X$) implies that $F'$ also certifies that $(X,\mathcal{M}, i, S)\in \mathcal{B}_t$.
If there is an $\mathcal{M}'\in \mathsf{add}(\mathcal{M},v)$ such that $(X\cup\{v\},\mathcal{M}', i,S)\in \mathcal{F}_{t'}$, then
there is a linear forest $F'$ certifying that $(X\cup\{v\},\mathcal{M}',i,S)\in \mathcal{B}_{t'}$.
In this case, $V(F')\cap X_t  = X$ and $\mathsf{sig}_{G_{t},X_{t}}(F')=\mathcal{M}$.
Therefore, $F'$ certifies that $(X,\mathcal{M}, i, S)\in \mathcal{B}_t$.
Finally, if $S=S'\cup\{v\}$ and there is an $\mathcal{M}'\in \mathsf{add}(\mathcal{M},v)$ such that $(X\cup\{v\},\mathcal{M}', i-1, S')\in \mathcal{F}_{t'}$ and $S'\cup\{v\}\in \mathcal{I}(M)$,
there is a linear forest $F'$ that certifies that $(X\cup\{v\},\mathcal{M}', i-1, S')\in \mathcal{B}_{t'}$.
The same linear forest $F'$ certifies that $(X,\mathcal{M}, i, S)\in \mathcal{B}_t$.\medskip

\noindent{\em Case 4:} $X_t$ is a join node.\medskip

Let $t_1$ and $t_2$ be the two children of $t$ and
let $(X,\mathcal{M}, i, S)$ such that $S\in \mathcal{S}_t [X,\mathcal{M},i]$.
By definition, there exist a pair
$(\mathcal{M}_1,\mathcal{M}_2)\in \xi (\mathcal{M})$
and 
two integers $i_1, i_2\in \{0,\ldots, k\}$
such that $i_1+i_2= i$ and,
if $X_i=U(\mathcal{M}_i),i\in\{1,2\}$, then $(X_1, \mathcal{M}_1, i_1, S_1)\in \mathcal{F}_{t_1}$, $(X_2, \mathcal{M}_2, i_2, S_2)\in \mathcal{F}_{t_2}$, and $S=S_1\cup S_2\in \mathcal{I}(M)$.
By induction hypothesis, there is a linear forest $F_1\subseteq G_{t_1}$ certifying that $(X_1, \mathcal{M}_1, i_1, S_1)\in \mathcal{B}_{t_1}$ and a linear forest $F_2\subseteq G_{t_2}$ certifying that $(X_2, \mathcal{M}_2, i_2, S_2)\in \mathcal{B}_{t_2}$.
Since $\mathcal{M}_1\cup \mathcal{M}_2 = \mathcal{M}$ and $\mathcal{M}_1\cap \mathcal{M}_2 = \emptyset$, it holds that $F_1\cup F_2$ is a linear forest of $G_t$ such that $\mathsf{sig}_{G_t,X_t}(F_1\cup F_2) = \mathcal{M}$.
Therefore, we get that $F_1\cup F_2$ certifies that $(X,\mathcal{M}, i, S)\in \mathcal{B}_t$.
\end{proof}

We now prove the following lemma which shows the correctness of our dynamic programming algorithm.

\begin{lemma}\label{lem:dp-invariant}
Let $P$ be an $(s,t)$-path of $G$ and let $S$ be a subset of $V(P)$ that is an independent set of $M$ of size at least $k$.
For every $t\in V(T)$,
if $F_t$ is the graph $P\cap G_t$, $S_t = S\cap V(G_t \setminus X_t)$, $|S_t| = i$, and
$X=V(P)\cap X_t$, and $\mathcal{M}=\mathsf{sig}_{G_t,X_t}(F_t)$,
then for every $S'\in \mathcal{S}_t[X,\mathcal{M}, i]$, if $F'$ certifies that  $(X,\mathcal{M},i,S')\in \mathcal{B}_t$,  then $F'\cup (P\setminus V(F_t))$ is an $(s,t)$-path of $G$ that contains the set $S'\cup (S\setminus S_t)$.
\end{lemma}

\begin{proof}
We prove the lemma by bottom-up induction on the decomposition tree.
Let $P$ be an $(s,t)$-path of $G$ and let $S$ be a subset of $V(P)$ that is an independent set of $M$ of size at least $k$.
Also, let $t\in V(T)$.
Let $F_t$ be the graph $P\cap G_t$, $S_t = S\cap V(G_t\setminus X_t)$, $|S_t| = i$, and
$X=V(P)\cap X_t$, and $\mathcal{M}=\mathsf{sig}_{G_t,X_t}(F_t)$.
\medskip

\noindent{\em Case 1}:
$X_t$ is an leaf node.
\medskip

In this case, the lemma holds trivially.
\medskip

\noindent{\em Case 2}:
$X_t$ is an insert node.
\medskip

Let $t'$ be the child of $t$.
By induction hypothesis,
if  $F_{t'}$ is the graph $P\cap G_{t'}$, $S_{t'}= S\cap V(G_{t'}\setminus X_{t'})$, $|S_{t'}| = i$,
$X'=V(P)\cap X_{t'}$, and $\mathcal{M}'=\mathsf{sig}_{G_{t'},X_{t'}}(F_{t'})$,
then if $S''\in \mathcal{S}_{t'}[X',\mathcal{M}',i]$ such that $(X',\mathcal{M}',i,S'')\in \mathcal{B}_{t'}$ and $F''$ certifies that  $(X',\mathcal{M}',i,S'')\in \mathcal{B}_{t'}$ then $F''\cup (P\setminus V(F_{t'}))$ is an $(s,t)$-path of $G$ that contains the set $S''\cup (S\setminus S_{t'})$.
We will prove that for every $S'\in \mathcal{S}_t [X,\mathcal{M},i]$, if $F'$ certifies that  $(X,\mathcal{M},i,S')\in \mathcal{B}_{t}$, then $F'\cup (P\setminus V(F_{t}))$ is an $(s,t)$-path of $G$ that contains the set $S'\cup (S\setminus S_{t})$.

If $v\notin X$, then the fact that $X_t = X_{t'}\cup \{v\}$ implies that $F_t= F_{t'}$, $S_t = S_{t'}$, $X=X'$ and $\mathcal{M} = \mathcal{M}'$. Therefore, since in this case
$\mathcal{S}_{t}[X,\mathcal{M},i] = \{S\mid (X, \mathcal{M}, i,S)\in \mathcal{F}_{t'}\}$,
it holds that $(X, \mathcal{M}, i,S')\in \mathcal{F}_{t'}$.
Also, by~\autoref{lem:dp-partial}, $\mathcal{F}_{t'}\subseteq \{(X,\mathcal{M}, i, S)\mid S\in \mathcal{S}_{t'} [X,\mathcal{M},i]\}\subseteq \mathcal{B}_{t'}$.
Therefore, if $F'$ certifies that $(X, \mathcal{M}, i,S')\in \mathcal{B}_{t'}$, then $F'\cup (P\setminus V(F_{t'}))$ is an $(s,t)$-path of $G$ that contains the set $S'\cup (S\setminus S_{t'})$.
Observe that $F'$ also certifies that $(X, \mathcal{M}, i,S')\in \mathcal{B}_t$ and since $F_t= F_{t'}$ and $S_t = S_{t'}$, we have $F'\cup (P\setminus V(F_{t})) = F'\cup (P\setminus V(F_{t'}))$ and $S'\cup (S\setminus S_{t}) = S'\cup (S\setminus S_{t'})$.
Thus, $F'\cup (P\setminus V(F_{t}))$ is an $(s,t)$-path of $G$ that contains the set $S'\cup (S\setminus S_{t})$.

If $v\in X$ and $\mathcal{M}^{(v)}_2\subseteq E(G_t)$, then observe that $F_{t'} =  F_{t}\setminus \{v\}$, $S_{t} = S_{t'}$ (since $S_t = S\cap V(G_t\setminus X_t) = S\cap V(G_{t'}\setminus X_{t'})$),
$X'= X\setminus \{v\}$,
and 
$\mathcal{M}'= \mathsf{rem}(\mathcal{M},v)$.
Therefore,
since $\mathcal{S}_t [X,\mathcal{M},i] =
\{S\mid (X\setminus \{v\}, \mathsf{rem}(\mathcal{M}, v), i,S)\in \mathcal{F}_{t'}\}$,
and $S'\in \mathcal{S}_t [X,\mathcal{M},i]$,
we have that $(X',\mathcal{M}',i,S')\in \mathcal{F}_{t'}$.
Also, by~\autoref{lem:dp-partial}, $\mathcal{F}_{t'}\subseteq \{(X,\mathcal{M}, i, S)\mid S\in \mathcal{S}_{t'} [X,\mathcal{M},i]\}\subseteq \mathcal{B}_{t'}$.
Therefore, 
if $F''$ certifies that $(X', \mathcal{M}', i,S')\in \mathcal{B}_{t'}$, then $F''\cup (P\setminus V(F_{t'}))$ is an $(s,t)$-path of $G$ that contains the set $S'\cup (S\setminus S_{t'})$.
Let $F' = F''\cup (\{v\},\mathcal{M}^{(v)}_2)$.
Notice that $F'$ is a linear forest and this follows from the fact that $F_{t}$ and $F''$ are linear forests, $F_{t'} = F_{t}\setminus \{v\}$, $\mathsf{sig}_{G_{t'},X_{t'}}(F_{t'}) =\mathsf{sig}_{G_{t'},X_{t'}}(F'')$, and $\mathcal{M}^{(v)}_2\subseteq E(G_t)$.
Therefore,
$F'$ certifies that $(X,\mathcal{M}, i, S'')\in \mathcal{B}_t$.
Also, we have that $F'\cup (P\setminus V(F_{t})) = F''\cup (P\setminus V(F_{t'}))$ and $S''\cup (S\setminus S_{t}) = S''\cup (S\setminus S_{t'})$.
Thus, $F'\cup (P\setminus V(F_{t}))$ is an $(s,t)$-path of $G$ that contains the set $S''\cup (S\setminus S_{t})$.
To conclude Case 2, observe that
if $v\in X$ and $E(G_t)\setminus \mathcal{M}^{(v)}_2\neq \emptyset$, $\mathcal{S}_t[X,\mathcal{M},i] = \emptyset$.
\medskip

\noindent{\em Case 3}:
$X_t$ is a forget node.
\medskip

Let $t'$ be the child of $t$ and let $v$ be the vertex in $X_{t'}\setminus X_t$.
By induction hypothesis,
if $F_{t'}$ is the graph $P\cap G_{t'}$, $S_{t'}= S\cap V(G_{t'}\setminus X_{t'})$, $|S_{t'}| = i$,
$X'=V(P)\cap X_{t'}$, and $\mathcal{M}'=\mathsf{sig}_{G_{t'},X_{t'}}(F_{t'})$,
then if $S''\in \mathcal{S}_{t'}[X',\mathcal{M}',i]$ and $F''$ certifies that  $(X',\mathcal{M}',i,S'')\in \mathcal{B}_{t'}$ then $F''\cup (P\setminus V(F_{t'}))$ is an $(s,t)$-path of $G$ that contains the set $S''\cup (S\setminus S_{t'})$.
Let $S'\in \mathcal{S}_t [X,\mathcal{M},i]$.

If $(X,\mathcal{M}, i, S')\in \mathcal{F}_{t'}$, then,
by \autoref{lem:dp-partial},
$(X,\mathcal{M}, i, S')\in \mathcal{B}_{t'}$, and therefore there is
a linear forest $F''\subseteq G_{t'}$ certifying that
$(X,\mathcal{M}, i, S')\in \mathcal{B}_{t'}$.
Observe that since $V(F'')\cap X_{t'} = X$, we have that
$F''$ is also a linear forest in $G_{t}$ certifying that $(X,\mathcal{M}, i, S')\in \mathcal{B}_{t}$.
Therefore, since $F''\cup (P\setminus V(F_{t'})) = F''\cup (P\setminus V(F_{t}))$ and $S''\cup (S\setminus S_{t}) = S''\cup (S\setminus S_{t'})$, we have
that 
$F''\cup (P\setminus V(F_{t}))$ is an $(s,t)$-path of $G$ that contains the set $S'\cup (S\setminus S_{t})$.

If there is an $\mathcal{M}'\in \mathsf{add}(\mathcal{M},v)$ such that $(X\cup\{v\},\mathcal{M}', i,S')\in \mathcal{F}_{t'}$, then,
by \autoref{lem:dp-partial}, $(X\cup\{v\},\mathcal{M}', i,S')\in \mathcal{B}_{t'}$ and therefore there is
a linear forest $F''\subseteq G_{t'}$ certifying that
$(X\cup\{v\},\mathcal{M}', i,S')\in \mathcal{B}_{t'}$.
Notice that, since $X_t = X_{t'}\setminus \{v\}$,
we have that $V(F'')\cap X_{t}$ and $\mathsf{sig}_{G_t,X_{t}}(F'') = \mathcal{M}$.
Thus, $F''$ certifies that $(X,\mathcal{M}, i,S')\in \mathcal{B}_{t}$.
Moreover,
since $F''\cup (P\setminus V(F_{t'})) = F''\cup (P\setminus V(F_{t}))$ and $S''\cup (S\setminus S_{t}) = S''\cup (S\setminus S_{t'})$, we have
that 
$F''\cup (P\setminus V(F_{t}))$ is an $(s,t)$-path of $G$ that contains the set $S'\cup (S\setminus S_{t})$.

If $S' = S''\cup\{v\}$ and there is an $\mathcal{M}'\in \mathsf{add}(\mathcal{M},v)$ such that $(X\cup\{v\},\mathcal{M}', i-1, S'')\in \mathcal{F}_{t'}$ and $S'\in \mathcal{I}(M)$, then
by \autoref{lem:dp-partial}, $(X\cup\{v\},\mathcal{M}', i-1,S'')\in \mathcal{B}_{t'}$ and therefore there is
a linear forest $F''\subseteq G_{t'}$ certifying that
$(X\cup\{v\},\mathcal{M}', i-1,S'')\in \mathcal{B}_{t'}$.
The fact that $X_t = X_{t'}\setminus \{v\}$
implies that $S'\subseteq V(G_t\setminus X_t)$,
$V(F'')\cap X_{t}$ and $\mathsf{sig}_{G_t,X_{t}}(F'') = \mathcal{M}$.
Therefore, $F''$ certifies that $(X,\mathcal{M}, i,S')\in \mathcal{B}_{t}$.
Moreover,
since $F''\cup (P\setminus V(F_{t'})) = F''\cup (P\setminus V(F_{t}))$ and $S''\cup (S\setminus S_{t}) = S''\cup (S\setminus S_{t'})$, we have
that 
$F''\cup (P\setminus V(F_{t}))$ is an $(s,t)$-path of $G$ that contains the set $S'\cup (S\setminus S_{t})$.
This concludes Case 3.
\medskip

\noindent{\em Case 4:} $X_t$ is a join node.\medskip

Let $t_1, t_2$ be the two children of $t$ and assume that the induction hypothesis holds for both $t_1,t_2$.
Keep in mind that $X_{t_1} = X_{t_2} = X_t$.
Also, let $S'\in \mathcal{S}_t [X,\mathcal{M},i]$.
By definition, $S' = S_1\cup S_2$ such that there is a pair $(\mathcal{M}_1, \mathcal{M}_2)\in \xi(\mathcal{M})$ and two integers $i_1, i_2\in \{0,\ldots, k\}$ such that $i_1 +i_2 = i$, and if  $X_i=U(\mathcal{M}_i),i\in\{1,2\}$, then $S_1\cup S_2\in \mathcal{I}(M)$,
$(X_1, \mathcal{M}_1, i_1, S_1)\in \mathcal{F}_{t_1}$
and $(X_2, \mathcal{M}_2, i_2, S_2)\in \mathcal{F}_{t_2}$.
Due to~ \autoref{lem:dp-partial}, $(X_1, \mathcal{M}_1, i_1, S_1)\in \mathcal{B}_{t_1}$ and $(X_2, \mathcal{M}_2, i_2, S_2)\in \mathcal{B}_{t_2}$, and therefore there are linear forests $F_{1}\subseteq G_{t_1}$ and $F_2\subseteq G_{t_2}$ such that for every $i\in\{1,2\}$, $F_i$ certifies that
$(X_i, \mathcal{M}_i, i_i, S_i)\in \mathcal{B}_{t_i}$.
Moreover, by induction hypothesis, $F_i\cup (P\setminus V(F_{t_i}))$ is an $(s,t)$-path of $G$ that contains the set $S_i\cup (S\setminus S_{t_i})$.
The fact that $(\mathcal{M}_1, \mathcal{M}_2)\in \xi(\mathcal{M})$ implies that $F_1\cup F_2$ is a linear forest of $G_{t}$ such that $V(F_1\cup F_2)\cap X_t = X$ and $\mathsf{sig}_{G_t,X_t}(F_1\cup F_2) = \mathcal{M}$.
Moreover, since $S_1\cup S_2\in \mathcal{I}(M)$  and $i_1 +i_2 = i$, we get that $F_1\cup F_2$ certifies that $(X,\mathcal{M}, i, S_1\cup S_2)\in \mathcal{B}_t$.
Also, the fact that for every $i\in\{1,2\}$, $F_i\cup (P\setminus V(F_{t_i}))$ is an $(s,t)$-path of $G$ that contains the set $S_i\cup (S\setminus S_{t_i})$ implies
that $F_1\cup F_2 \cup (P\setminus V(F_t))$ is an $(s,t)$-path of $G$ that contains the set $S_1\cup S_2\cup (S\setminus V(S_{t}))$.
This concludes Case 4.
\end{proof}

We conclude this subsection by proving \autoref{lem:dp-planar}.

\begin{proof}[Proof of \autoref{lem:dp-planar}]
We first observe that we can transform a given tree decomposition of width $q$ to a nice tree decomposition $(T,\chi)$ of width $q$ in time $\mathcal{O}(q^2\cdot n)$.
Moreover, $|V(T)| = q^{\mathcal{O}(1)}\cdot n$.
Then, for every $t\in V(T)$, we compute the set $\mathcal{F}_t$ in a bottom-up way.
By \autoref{lem:dp-partial}, $|\mathcal{F}_t|= 2^{q^{\mathcal{O}(1)}}$
and each $\mathcal{F}_t$ can be computed in time $2^{q^{\mathcal{O}(1)}}\cdot (|G|+\|M\|)^{\mathcal{O}(1)}$, resulting to the claimed overall running time.
If there is a set $S\subseteq V(G)$ such that $\{\emptyset,\{\emptyset\}, k, S\}\in \mathcal{F}_r$, then we output a report that $G$ contains an $(s,t)$-path of rank at least $k$, otherwise we report
that such a path does not exist in $G$.
The correctness of the algorithm follows from~\autoref{obs:solut},~\autoref{lem:dp-invariant} and the fact that, by~\autoref{thm:rep-det}, for every $t\in V(T)$, $\mathcal{F}_t$ is a $k$-representative family and $\mathcal{F}_t\subseteq \{(X,\mathcal{M}, i,S)\mid S\in \mathcal{S}_t[X,\mathcal{M},i]\}$.
\end{proof}

\subsection{Proof of \autoref{thm:fpt-planar}}\label{subsec:proofThm2}
In the proof of \autoref{thm:fpt-planar}, we will use the single-exponential time 2-approximation algorithm for treewidth of Korhonen~\cite{Korhonen21tw}.

\begin{proposition}\label{prop:twaprrox}
There exists an algorithm that given a graph $G$ and an integer $k\in \mathbb{N}$,
outputs, in time $2^{\mathcal{O}(k)}\cdot |G|$, either a tree decomposition of $G$ of width at most $2k+1$ or a report that the treewidth of $G$ is larger than $k$.
\end{proposition}

\begin{proof}[Proof of~\autoref{thm:fpt-planar}]
Let $(G,M)$ be a framework, where $G$ is a planar graph and $M$ is a linear matroid given by its representation over a finite filed or the field of rationals, and let $k\in \mathbb{N}$.
We set $q=g(k)$, where $g$ is the function of \autoref{lem:reducing-tw}.
Keep in mind that $g(k)=2^{\mathcal{O}(k\log k)}$.
We describe an algorithm $\mathcal{A}$ that solves \textsc{Maximum Rank  $(s,t)$-Path}.

Our algorithm $\mathcal{A}$ first calls the algorithm of \autoref{prop:twaprrox} for $G$ and $q$ which runs in time $2^q\cdot n =2^{2^{\mathcal{O}(k\log k)}}\cdot n$ and outputs either a tree decomposition of $G$ of width at most $2q$ or a report that the treewidth of $G$ is larger than $q$.
In the first possible output, we use the algorithm of \autoref{lem:dp-planar}, which runs in time $2^{q^{\mathcal{O}(1)}}\cdot (|G|+\|M\|)^{\mathcal{O}(1)} = 2^{2^{\mathcal{O}(k\log k)}}\cdot (|G|+\|M\|)^{\mathcal{O}(1)}$, and we solve \textsc{Maximum Rank  $(s,t)$-Path}.
In the second possible output (i.e., where $G$ has treewidth at least $q$), we apply the algorithm of \autoref{lem:reducing-tw} and, in time $2^{2^{\mathcal{O}(k\log k)}}\cdot (|G|+\|M\|)^{\mathcal{O}(1)}$, we either report a positive answer to \textsc{Maximum Rank  $(s,t)$-Path} or find a vertex $v\in V(G)$ such that $(G,M,k,s,t)$ and $(G\setminus v, M\setminus v,k,s,t)$ are equivalent instances of \textsc{Maximum Rank  $(s,t)$-Path}.
If the latter happens, we recursively run $\mathcal{A}$ for the framework $(G\setminus v, M\setminus v)$.
Observe that the overall running time of $\mathcal{A}$ is $2^{2^{\mathcal{O}(k\log k)}}\cdot (|G|+\|M\|)^{\mathcal{O}(1)}$.
\end{proof}

\section{Conclusion}\label{sec:conclusion}
In this paper, we provide a deterministic FPT algorithm
for \textsc{Maximum Rank $(s,t)$-Path} for frameworks $(G,M)$,
where $G$ is a planar graph
and $M$ is represented over a finite field or the rationals.
Let us conclude by discussing some open research directions.

Since the algorithm of~\cite{FominGKSS23fixe} for \textsc{Maximum Rank $(s,t)$-Path} runs in time $2^{\mathcal{O}(k^2\log (k+q))}n^{\mathcal{O}(1)}$, a natural question is whether one can drop the double-exponential dependence on the parameter $k$ on the running time of the algorithm of \Cref{thm:fpt-planar}.
The main bottleneck is the bound the treewidth of a graph that contains no irrelevant vertices.
In particular, our approach to detect irrelevant vertices requires a recursive zooming into a given wall of the graph in order to find a packing of $k+1$-many $k$-walls with compasses of specific rank.
To perform this zooming, one should ask for the initial wall to be of height at least $k^{\mathcal{O}(k)}$.
It is unclear whether we can circumvent this argument and detect irrelevant vertices if the initial wall has height linear (or even polynomial) in $k$.

As mentioned in the introduction, the method of~\cite{FominGKSS23fixe} gives a randomized algorithm for the more general problem of \textsc{Maximum Rank $(S,T)$-Linkage}. In this paper, we focus on the special case where $|S|=|T|=1$ and one could ask whether our techniques can be applied to solve the general problem of detecting $(S,T)$-linkages of large rank for frameworks with planar graphs and matroids represented over finite fields.
Such a generalization of our results does not seem to be trivial and therefore we leave this as an open research direction.

Another natural question to ask is whether our approach can be generalized to obtain deterministic FPT algorithms for frameworks with more general classes of graphs. While it seems plausible to extend the applicability of the irrelevant vertex technique arguments up to graphs that exclude a graph as a minor, such a proof would be highly technical. For frameworks with general graphs, it is very unclear whether one can achieve rerouting that does not decrease the rank and therefore allow an irrelevant vertex argument to go through.

Also, in the lines of~\cite{FominGKSS23fixe}, an interesting open question is whether we can obtain a deterministic FPT algorithm for \textsc{Maximum Rank $(s,t)$-Path} for frameworks with matroids not representable in finite fields of small order or in the field of rationals.
For example, uniform matroids, and more generally transversal matroids, are representable over a finite field, but the field of representation must be large enough.
While the approach of~\cite{FominGKSS23fixe} also gives a \textit{randomized} FPT algorithm for frameworks of transversal matroids, our dynamic programming subroutine relies on the efficient computation of representative sets, which requires a linear representation of the input matroid. We stress that this is the only place in the proof of \Cref{thm:fpt-planar} requiring a linear representation of the matroid.
Another interesting open question, is whether \textsc{Maximum Rank $(s,t)$-Path} is \classFPT when parameterized by $k$ and the treewidth if the input matroid is given by its independence oracle.

\bibliographystyle{siam}

\end{document}